\documentclass[11pt,a4paper]{amsart}
\usepackage[foot]{amsaddr}
\usepackage[T1]{fontenc}

\usepackage{ifxetex}
\ifxetex
  \usepackage[no-math]{fontspec}
\else
\fi
\usepackage{amsmath}
\usepackage{amsfonts}
\usepackage{amssymb}
\usepackage{amsthm}
\usepackage{fullpage}
\usepackage{tablefootnote}
\usepackage{hyphenat}
\usepackage{microtype}
\RequirePackage{silence}
\usepackage{hyphenat}
\ifxetex
  \usepackage[libertine]{newtxmath}
\else
  \usepackage{newtxmath}
\fi
\usepackage[tt=false]{libertine}
\usepackage{caption}
\usepackage{tcolorbox}
\tcbuselibrary{breakable,skins}
\usepackage{bbm}
\usepackage{hyperref, color}
\hypersetup{colorlinks=true,citecolor=blue, linkcolor=blue, urlcolor=blue}
\usepackage[linesnumbered,boxed,ruled,vlined]{algorithm2e}
\usepackage{bm}
\usepackage{bbm}
\usepackage[numbers]{natbib}
\usepackage{xcolor}
\usepackage{enumerate}
\usepackage{enumitem}
\usepackage{ragged2e}
\usepackage{tabularx}
\usepackage{array}
\usepackage{longtable}
\usepackage{hhline}
\usepackage{multirow}
\newcolumntype{L}[1]{>{\=Raggedright\arraybackslash}p{#1}}
\newcolumntype{C}[1]{>{\centering\arraybackslash}m{#1}}
\newcolumntype{R}[1]{>{\=Raggedleft\arraybackslash}p{#1}}

\usepackage{makecell}
\usepackage{prettyref}
\usepackage{footnote}
\usepackage{float}
\usepackage{fix-cm}
\makesavenoteenv{tabular}

\makeatletter
\renewcommand\paragraph{\@startsection{paragraph}{4}\z@\z@{-\fontdimen2\font}{\normalfont\itshape}}
\makeatother

\renewcommand{\epsilon}{\varepsilon}

\newtheorem{theorem}{Theorem}[section]
\newcommand{\newsharedtheorem}[2]{\newtheorem{#1}[theorem]{#2}}
\newsharedtheorem{observation}{Observation}
\newsharedtheorem{claim}{Claim}
\newtheorem*{claim*}{Claim}
\newsharedtheorem{condition}{Condition}
\newsharedtheorem{example}{Example}
\newsharedtheorem{fact}{Fact}
\newsharedtheorem{lemma}{Lemma}
\newsharedtheorem{proposition}{Proposition}
\newsharedtheorem{conjecture}{Conjecture}
\newsharedtheorem{corollary}{Corollary}
\theoremstyle{definition}

\newsharedtheorem{definition}{Definition}
\newsharedtheorem{remark}{Remark}
\newtheorem*{remark*}{Remark}
\newtheorem{assumption}{Assumption}

\newrefformat{thm}{Theorem~\ref{#1}}
\newrefformat{cond}{Condition~\ref{#1}}
\newrefformat{cor}{Corollary~\ref{#1}}
\newrefformat{def}{Definition~\ref{#1}}
\newrefformat{definition}{Definition~\ref{#1}}
\newrefformat{lem}{Lemma~\ref{#1}}
\newrefformat{lemma}{Lemma~\ref{#1}}
\newrefformat{Alg}{Algorithm~\ref{#1}}
\newrefformat{observation}{Observation~\ref{#1}}
\newrefformat{claim}{Claim~\ref{#1}}
\newrefformat{example}{Example~\ref{#1}}
\newrefformat{fact}{Fact~\ref{#1}}
\newrefformat{proposition}{Proposition~\ref{#1}}
\newrefformat{conjecture}{Conjecture~\ref{#1}}
\newrefformat{remark}{Remark~\ref{#1}}
\newrefformat{assumption}{Assumption~\ref{#1}}

\let\pref\prettyref
\newcommand{\Cref}[1]{\pref{#1}}

\newcommand{\poly}{\textnormal{poly}}

\newcommand{\abs}[1]{\left\vert#1\right\vert}
\newcommand{\set}[1]{\left\{#1\right\}}
\newcommand{\tuple}[1]{\left(#1\right)}

\newcommand{\tp}{\tuple}

\newcommand{\NP}{\ensuremath{\mathbf{NP}}}

\newcommand{\RP}{\ensuremath{\mathbf{RP}}}
\newcommand{\defeq}{\triangleq}

\newcommand{\fr}{\mathsf{free}}
\newcommand{\frin}{\mathsf{free}^{(1)}}

\newcommand{\vbl}{\textnormal{vbl}}

\newcommand{\tr}{\mathsf{Tree}}

\newcommand{\rrat}{\mathsf{RandMargEst}}
\newcommand{\prop}{\mathsf{Propose}}

\newcommand{\reci}{\mathsf{RecipEst}}

\newcommand{\e}{\mathrm{e}}

\def\*#1{\boldsymbol{#1}}
\def\+#1{\mathcal{#1}}
\def\-#1{\mathrm{#1}}
\def\=#1{\mathbb{#1}}

\DeclareMathOperator{\oPr}{\mathbb{P}}
\DeclareMathOperator{\opr}{\mathbf{Pr}}
\renewcommand{\Pr}[2][]
{\ifthenelse{\isempty{#1}}
  {\oPr\left[#2\right]}
  {\oPr_{#1}\left[#2\right]}
} \newcommand{\pr}[2][]
{\ifthenelse{\isempty{#1}}
  {\opr\left[#2\right]}
  {\opr_{#1}\left[#2\right]}
}

\DeclareMathOperator{\oE}{\mathbb{E}}
\newcommand{\E}[2][]
{\ifthenelse{\isempty{#1}}
  {\oE\left[#2\right]}
  {\oE_{#1}\left[#2\right]}
} 
\DeclareMathOperator{\oVar}{\mathrm{Var}}
\newcommand{\Var}[2][]
{\ifthenelse{\isempty{#1}}
  {\oVar\left[#2\right]}
  {\oVar_{#1}\left[#2\right]}
}
\def\oEnt{\mathrm{Ent}}
\NewDocumentCommand{\Ent}{ O{} O{} m }{
  \ifthenelse{\isempty{#1}} {
    \ifthenelse{\isempty{#2}} {
      \oEnt\left[#3\right]
    } {
      \oEnt^{#2}\left[#3\right]
    }
  } {
    \ifthenelse{\isempty{#2}} {
      \oEnt_{#1}\left[#3\right]
    } {
      \oEnt_{#1}^{#2}\left[#3\right]
    }
  }
}

\usepackage[textsize=tiny]{todonotes}

\usepackage{xifthen}

\newcommand{\True}{\mathtt{True}}
\newcommand{\False}{\mathtt{False}}

\newcommand{\dist}{\mathrm{dist}}
\newcommand{\one}[1]{\mathbf{1}\left[#1\right]}
\makeatletter
\def\prob#1#2#3{\goodbreak\begin{list}{}{\labelwidth\z@ \itemindent-\leftmargin
                        \itemsep\z@  \topsep6\p@\@plus6\p@
                        \let\makelabel\descriptionlabel}
                \item[\it Name]#1
               \item[\it Instance]                #2
                \item[\it Output]#3
                \end{list}}
\makeatother

\usepackage{mathtools,needspace,etoolbox}
\newif\ifmergemarks
\ifdefined\MergeClean\mergemarksfalse\else\mergemarkstrue\fi
\definecolor{MergeOriginal}{HTML}{000000}
\definecolor{SamplingRewrite}{HTML}{006B4F}
\ifmergemarks
  \definecolor{MergeAdded}{HTML}{124CA0}
  \definecolor{MergeChanged}{HTML}{8B287C}
\else
  \definecolor{MergeAdded}{HTML}{000000}
  \definecolor{MergeChanged}{HTML}{000000}
\fi

\hypersetup{
  pdftitle={A Counting and Sampling Lovasz Local Lemma},
  pdfauthor={Hongyang Liu, Chunyang Wang, Yitong Yin, Yiyao Zhang, Can Zhou},
  pdfsubject={Counting and sampling CSP solutions at the quadratic local lemma threshold},
  bookmarksdepth=2}
\newrefformat{sec}{Section~\ref{#1}}
\newrefformat{alg}{Algorithm~\ref{#1}}
\newrefformat{prop}{Proposition~\ref{#1}}
\newcommand{\CF}{\+F}

\newcommand{\Sol}{\Omega}

\newbool{doubleblind}
\setbool{doubleblind}{false}

\title[A Counting and Sampling Lov\'asz Local Lemma]{A Counting and Sampling Lov\'asz Local Lemma}

\ifdoubleblind
\author{Author(s)}
\else
\author{Hongyang Liu, Chunyang Wang, Yitong Yin,
  Yiyao Zhang, Can Zhou}

\address[Hongyang Liu, Yitong Yin, Yiyao Zhang, Can Zhou]
{School of Computer Science, State Key Laboratory for Novel Software
Technology, New Cornerstone Science Laboratory, Nanjing University,
Nanjing, China.
\textnormal{Email:
\texttt{liuhongyang@smail.nju.edu.cn,
yinyt@nju.edu.cn,
\{zhangyiyao,bzy.cirno\}@smail.nju.edu.cn}}}

\address[Chunyang Wang]
{National Institute of Informatics, Tokyo, Japan.
\textnormal{Email: \texttt{c\_wang@nii.ac.jp}}}

\fi

\makeatletter
\patchcmd{\@startsection}{\@m}{\@M}{}{\PackageError{tocfix}{Could not patch starred headings}{}}
\makeatother

\begin{document}

\begin{abstract}
We establish counting and sampling analogues of the Lov\'asz Local Lemma:
we give efficient algorithms for approximately counting and exactly
sampling satisfying assignments of general constraint satisfaction problems
(CSPs) in the local lemma regime
\[
4 \e p (D+1)^2\leq 1,
\]
where $p$ is the maximum constraint violation probability and $D$ is the maximum dependency degree.

This condition is tight up to constant factors under $\NP\neq\RP$, matching known hardness bounds for counting and sampling in natural subclasses of CSPs. 
Our key ingredient is a novel \emph{$2$-tree expansion} for constraint marginal probabilities that exhibits exponential decay
of correlations throughout this regime.

This expansion yields deterministic polynomial-time approximate counting
for fixed local parameters, randomized approximate counting with quadratic
cost, and exact sampling in expected near-linear time when the local
parameters are fixed.
\end{abstract}	

\hypersetup{pageanchor=false}
\maketitle

\setcounter{tocdepth}{1}
\tableofcontents

\thispagestyle{empty}
\newpage

\setcounter{page}{1}
\hypersetup{pageanchor=true}

\section{Introduction}

Constraint Satisfaction Problems (CSPs) provide a fundamental framework for studying combinatorial solution spaces and modeling a wide range of problems in computer science. 
A CSP consists of a collection of local constraints over a set of variables, each taking values from a finite domain, and a solution is an assignment that simultaneously satisfies all constraints. 
Determining whether a CSP instance admits a solution is one of the central problems in the theory of computation, leading to landmark results including~\cite{cook1971complexity,schaefer1978complexity,feder1998computational,
bulatov2017dichotomy,zhuk2020dichotomy}. 

A fundamental tool for proving the existence of solutions is the Lov\'asz Local Lemma (LLL)~\cite{LocalLemma}, which guarantees satisfiability even when constraints are individually unlikely to be violated but may exhibit limited dependencies. 
In its symmetric form, the LLL asserts that a CSP is satisfiable whenever
\[
\e p(D+1)\leq 1,
\]
where $p$ is the maximum violation probability of a constraint under a uniform random assignment, and $D$ is the maximum degree of the constraint dependency graph. 
Shearer~\cite{shearer85} later showed that this condition is essentially tight, establishing the sharp threshold for the existence of satisfying assignments in the local lemma regime.

{Beyond satisfiability, two fundamental problems in the study
of CSP solution spaces are \emph{approximate counting} and
\emph{sampling}. Approximate counting aims to estimate the
total number of solutions, while sampling aims to generate
a solution from the uniform or a close-to-uniform distribution over all solutions.
Together, they provide quantitative tools for understanding the size
of the solution space and the properties of typical solutions.
When CSPs are viewed as statistical physics systems with high-order
interactions, counting corresponds to estimating partition functions,
while sampling corresponds to generating equilibrium configurations.
Both problems also play a central role in statistical inference,
where counting provides normalizing constants and sampling enables
the estimation of marginal probabilities and expectations.}

{A central question in the computational complexity of
approximate counting and sampling is to characterize their
\emph{computational phase transitions}: under what conditions are
these problems tractable, and when do they become intractable?
Sharp computational phase transitions have been established for
several classes of spin systems~\cite{weitz06counting,sly2010computation,
li2013correlation,sly2014computational,sinclair2014approximation,
galanis2016inapproximability}, which correspond to weighted CSPs with
pairwise constraints. In contrast, for CSPs with general multi-variable
constraints, the tractability thresholds for approximate counting
and sampling remain largely open.}

{This motivates the study of a \emph{counting and sampling
Lov\'asz Local Lemma}, whose goal is to characterize the tractable
regimes for approximate counting and sampling using local-lemma-type
conditions.}
The seminal works of Moitra~\cite{Moi19} and Guo, Jerrum, and
Liu~\cite{GJL19} launched the modern algorithmic study of approximate
counting and sampling under local-lemma-type conditions and have since
shaped a broad literature on these problems~\cite{HSZ19,
guo2019counting,FGYZ20,feng2021sampling,vishesh21towards,
vishesh21sampling,HSW21,he2022sampling,he2022counting,
qiu2022perfect,feng2022improved,galanis2023inapproximability,
wang2024sampling,mann2025approximate,barvinok2026computing}.
The two problems are closely connected: simulated annealing provides a
general framework for approximate counting via sampling~\cite{stefankovic2009adaptive}.

{On the algorithmic side, earlier algorithms for general
CSPs achieved conditions of the form $pD^5\lesssim 1$ for approximate
counting and sampling~\cite{he2022counting,he2022sampling},
due to He, Wang, and Yin. Here, $\lesssim$ suppresses constants and
lower order factors.}
For the subclass of atomic CSPs, where each constraint forbids a
single local assignment, Wang and Yin obtained
{algorithms for both problems} under the nearly tight
condition $pD^{2+o_q(1)}\lesssim 1$, where $q$ is the uniform domain
size~\cite{wang2024sampling}.
{For the prominent special case of $k$-SAT, with Boolean
domains ($q=2$), this condition specializes to
$pD^{4.82}\lesssim 1$. These results left a substantial gap between
the algorithmic bounds and the threshold suggested by hardness
results.}

On the hardness side, the condition $pD^2\lesssim 1$ is known to be
necessary for approximate counting. 
In particular, assuming $\NP\neq\RP$, the hardness reductions in \cite{BGGGS19,galanis2023inapproximability} 
rule out a general efficient approximate counting algorithm beyond the condition
\begin{equation}\label{eq:counting-LLL-lower-bound}
    pD^2<4\e^2.
\end{equation}
A short verification of this bound, together with its
sampling consequence via a standard self-reduction,
appears in Appendix~\ref{app:known-lower-bound}.
Thus, counting and sampling exhibit a computational barrier at the
$pD^2$ scale, distinct from the $pD$ scale governing satisfiability
in the local lemma regime.

{Together, these algorithmic bounds and hardness results
raise a fundamental question about the computational tractability
of counting and sampling for general CSPs:}
\begin{center}
    \emph{{Is $pD^2\lesssim 1$ the correct threshold
    for the computational tractability of approximate counting
    and sampling?}}
\end{center}

\subsection{Our Results}

{We answer the above question affirmatively by establishing a counting and sampling Lov\'asz Local Lemma for general CSPs that is tight up to constant factors.}
Before stating our main results, we introduce the notation and parameters used throughout the paper.

A CSP instance (or CSP formula) is denoted by $\Phi=(V,\+Q,\+C)$, where $V$ is a set of $n=|V|$ variables. 
Each variable $v\in V$ is associated with a finite domain $Q_v$, and the configuration space is $\+Q\triangleq \bigotimes_{v\in V}Q_v$. 
The set $\+C$ consists of $m=\abs{\+C}$ constraints, where each constraint $c\in\+C$ is a Boolean function 
\[c:\+Q_{\vbl(c)}\rightarrow\{\True,\False\},\]
where $\+Q_{\vbl(c)}\defeq\bigotimes_{v\in\vbl(c)}Q_v$ and $\emptyset\ne\vbl(c)\subseteq V$ is the \emph{scope} of $c$.

A configuration $\sigma\in\+Q$ is called a \emph{satisfying assignment} (or a \emph{solution}) if every constraint in $\+C$ is satisfied by $\sigma$.
The \emph{partition function} of $\Phi$, denoted by $Z_\Phi$, is the number of satisfying assignments of $\Phi$.

\enlargethispage{3pt}
We use the following parameters associated with $\Phi=(V,\+Q,\+C)$:
\begin{itemize}
    \item the \emph{domain size} $q\triangleq\max_{v\in V}\abs{Q_{v}}$;
    \item the \emph{width} $k\triangleq\max_{c\in \+{C}}\abs{ {\vbl}(c)}$;
    \item the \emph{dependency degree} 
    \[
    D\triangleq\max\limits_{c\in \+{C}}\abs{\{c'\in \+{C}\setminus \{c\}\mid \vbl(c)\cap \vbl(c')\neq\emptyset\}};
    \]
    \item the \emph{violation probability} \[
    p\triangleq\max\limits_{c\in \+{C}}\mathbb{P}[\neg c],
    \]
    where $\mathbb{P}$ denotes the uniform product measure on $\+Q$.
\end{itemize}
We assume throughout that $q\ge2$, $k\ge2$, and $D\ge1$, as the remaining degenerate cases are trivial.

We consider the standard oracle model for CSPs, where constraint functions are accessed through the following evaluation oracle. 
This abstraction is necessary because arbitrary constraint functions, especially those with unbounded scope sizes, may not admit an efficient explicit representation.

\begin{assumption}[Evaluation Oracle]\label{assumption:evaluation-oracle}
There is an \emph{evaluation oracle} for $\Phi=(V,\+Q,\+C)$ such that, given a constraint $c\in\+C$ and an assignment $\sigma\in\+Q_{\vbl(c)}$, the oracle returns whether $c$ is satisfied by $\sigma$.
\end{assumption}
Each evaluation-oracle call is charged one operation; the \emph{oracle
cost} is the number of such calls.  
All running-time bounds use the computational model specified in
Section~\ref{sec:computational-model}.

This oracle model is also commonly assumed in the study of the algorithmic
Lov\'asz Local Lemma, e.g.~in~\cite{moser2010constructive,harvey2020algorithmic}, where
constraint evaluations are treated as primitive operations.

The following condition defines the local lemma regime considered in this paper.

\begin{condition}[{Counting and Sampling LLL Condition}]
\label{cond:sym-counting-LLL}
    {The parameters of 
    $\Phi=(V,\+Q,\+C)$ satisfy}
    \begin{equation}\label{eq:sym-local-lemma-condition}
        4\e p(D+1)^2\leq 1.
    \end{equation}
\end{condition}

We give deterministic and randomized approximate counting
algorithms, together with an exact sampler, under
\Cref{cond:sym-counting-LLL}. All three algorithms use the
evaluation oracle in \Cref{assumption:evaluation-oracle}.

We first give a deterministic algorithm for estimating the number
of CSP solutions.

\begin{theorem}[{Deterministic approximate counting}]
\label{thm:fptas-counting-LLL}
    There exists a deterministic algorithm that,
    given a CSP formula $\Phi$ satisfying \Cref{cond:sym-counting-LLL},
    access to the evaluation oracle in \Cref{assumption:evaluation-oracle},
    and an accuracy parameter $0<\varepsilon<1$,
    outputs an estimate $\widehat Z_{\Phi}$ satisfying
    \[
        (1-\varepsilon)Z_{\Phi}
        \leq \widehat Z_{\Phi}
        \leq (1+\varepsilon)Z_{\Phi}.
    \]
    {Both the oracle cost and the running time are bounded
    by $(nD/\varepsilon)^{O(kD\log q)}$ in the worst case.}
\end{theorem}

{For fixed $k$, $D$, and $q$,
\Cref{thm:fptas-counting-LLL} gives a fully polynomial-time approximation
scheme. As in previous deterministic approximate counting algorithms
in the local lemma regime~\cite{Moi19,guo2019counting,vishesh21towards,
he2022counting,wang2024sampling,feng2025toward}, the exponent of the
running time depends on the local parameters.}

{Randomization gives an approximate counting algorithm
with fixed-exponent polynomial time.}

\begin{theorem}[{Randomized approximate counting}]
\label{thm:fpras-counting-LLL}
    There exists a randomized algorithm that,
    given a CSP formula $\Phi$ satisfying \Cref{cond:sym-counting-LLL},
    access to the evaluation oracle in \Cref{assumption:evaluation-oracle},
    and an accuracy parameter $0<\varepsilon<1$,
    outputs an estimate $\widehat Z_{\Phi}$ satisfying
    \[
        \pr{
            (1-\varepsilon)Z_{\Phi}
            \leq \widehat Z_{\Phi}
            \leq (1+\varepsilon)Z_{\Phi}
        }
        \geq \frac{3}{4}.
    \]
The algorithm has oracle cost $\poly(k,D,q)(n/\varepsilon)^2$ and running time $\widetilde O\!\left(\poly(k,D,q)(n/\varepsilon)^2\right)$, both in the worst case. Here $\widetilde O$ suppresses
polylogarithmic factors in $m$, $n$, $q$, and $1/\varepsilon$.
\end{theorem}

Finally, we give an exact sampler under the same condition, using a linear expected number of constraint evaluations.

\begin{theorem}[Near-Linear time exact sampling]
\label{thm:fpaus-sampling-LLL}
    There exists a randomized algorithm that,
    given a CSP formula $\Phi$ satisfying \Cref{cond:sym-counting-LLL},
    and access to the evaluation oracle in
    \Cref{assumption:evaluation-oracle}, outputs a uniformly random
    satisfying assignment $X$.
    The algorithm has expected oracle cost $O(m)$
    and expected running time
    $\widetilde O\!\left(n+kmD^2\right)$.
    Here $\widetilde O$ suppresses
polylogarithmic factors in $m$, $n$, and $q$.
\end{theorem}

Note that $m\le n(D+1)$.
Thus the sampler has near-linear cost for fixed local parameters.
Its bounds hold in expectation, whereas the counting bounds hold
in the worst case.

Together with the hardness bound in
\eqref{eq:counting-LLL-lower-bound}, these results establish the optimal
$pD^2$ scale for approximate counting and sampling of general CSP solutions, up to
constant factors. 

\subsection{Technical Overview}

Our algorithms estimate the incremental effect of adding one constraint.
For $\+C$ and $c_0\in\+C$, define
\begin{equation}\label{eq:marginal-probability-def-technical-overview}
r_{\+C,c_0}
\defeq
\Pr{\neg c_0\mid\+C\setminus\{c_0\}}
=1-\frac{Z(\+C)}{Z(\+C\setminus\{c_0\})}.
\end{equation}
For an input order $\+C=\{c_1,\ldots,c_m\}$ and
$\+C_i=\{c_1,\ldots,c_i\}$,
\begin{equation}\label{eq:self-reduction-technical-overview}
Z_\Phi
=\left(\prod_{v\in V}|Q_v|\right)
\prod_{i=1}^{m}\bigl(1-r_{\+C_i,c_i}\bigr).
\end{equation}
Thus, approximate counting reduces to estimating constraint marginals.
Deleting constraints only decreases the dependency degree, so every
subinstance $\+C_i$ satisfies the same local-lemma condition.

\subsubsection{The $pD^2$ Threshold and $2$-Trees}

A $2$-tree is an independent set in the dependency graph $G(\+C)$
connected in $G^2(\+C)$.  Such distance-two structures originated in
Alon's algorithmic local lemma work~\cite{alon1991parallel} and have
appeared in counting and sampling in the local lemma regime~\cite{HSZ19,
Moi19,guo2019counting,FGYZ20,feng2021sampling,vishesh21towards,
vishesh21sampling,HSW21,he2022sampling,he2022counting,
wang2024sampling,feng2025toward}.

A $2$-tree captures the relevant probabilistic and combinatorial
scales simultaneously.  Its constraints have disjoint scopes, so a
$2$-tree $T$ contributes at most
\[
\Pr{\bigwedge_{c\in T}\neg c}\le p^{|T|}.
\]
On the other hand, at most $(\e D^2)^{t-1}$ rooted $2$-trees have size
$t$.  Hence their total scale is
$p^t(\e D^2)^{t-1}=p(\e pD^2)^{t-1}$.
This explains why $pD^2$ is the natural scale for correlation decay and
approximate counting.

This scale is not merely an upper-bound heuristic.  In the reduction
from the hard-core model to monotone CNF counting
in~\cite{BGGGS19}, vertices are replaced by blocks of Boolean
variables and edges by constraints on the corresponding blocks.  The
hard-core barrier beyond the uniqueness threshold~\cite{sly2014computational}
then becomes the LLL scale $pD^2\gtrsim1$.  These hard CSPs are blocked
and locally tree-like; on related structured instances,
\cite{wang2024sampling} attains the optimal scale by exploiting their
$2$-tree structure.  The challenge is to identify and use the same
structure in arbitrary CSPs.

Conceptually, our construction is a canonical cluster-like expansion.
Unlike a usual cluster expansion, it uses independent constraint sets
connected at distance two.  This retains the $p^{|T|}$ weight of
independent violations while recording their $D^{2|T|}$ combinatorial
growth.  It connects the algorithmic local lemma to the broader
cluster-expansion perspective in statistical physics
\cite{Kotecky1986cluster,helmuth2020algorithmic,cannon2020counting,
jenssen2020algorithms,borgs2022potts}.

\subsubsection{Canonical \texorpdfstring{$2$}{2}-Tree Expansion}

Our main counting idea is a canonical regrouping of
inclusion--exclusion.  The canonical extraction maps every constraint
set $\+D$ containing $c_0$ to a rooted $2$-tree $\tr(\+D)$.  Its key
property is that every fiber has the exact form
\[
\{\+D\subseteq\+C:c_0\in\+D,\ \tr(\+D)=T\}
=\{T\cup S:S\subseteq\fr(T)\},
\]
where $\fr(T)$ is the set of free constraints.  Summing
inclusion--exclusion within this fiber cancels the free choices and
leaves the single contribution
\[
(-1)^{|T|-1}
\Pr{\left(\bigwedge_{c\in T}\neg c\right)\land\fr(T)}.
\]

Telescoping the remaining normalization gives
\begin{equation}\label{eq:2-tree-expansion-recurrence-technical-overview}
r_{\+C,c_0}
=
\sum_{T\in\mathfrak T_{\+C,c_0}}
(-1)^{|T|-1}
\Pr{\bigwedge_{c\in T}\neg c\,\middle|\,\fr(T)}
\prod_{i=1}^{\ell^T}
\frac{1}{1-r_{\+C_i^T,c_i^T}},
\end{equation}
where $c_1^T,\ldots,c_{\ell^T}^T$ enumerate
$\+C\setminus(\{c_0\}\cup\fr(T))$, and
$\+C_i^T=\fr(T)\cup\{c_1^T,\ldots,c_i^T\}$.  This converts
inclusion--exclusion over arbitrary constraint sets into a structured
recursion over $2$-trees.  Every recursive call concerns a strict
subinstance.

\subsubsection{Correlation Decay and Approximate Counting}

A size-$t$ tree has a $p^t$ violation factor, while the conditional
local lemma controls the conditioning and reciprocal factors in
\eqref{eq:2-tree-expansion-recurrence-technical-overview}.  Together
with the rooted-tree bound, this gives a geometric tail for large
$2$-trees.  The same product bounds control the sensitivity of a tree
term to errors in its recursive marginals.  Although expanding a
size-$t$ tree leaves budget $L-t$ for its children, the sum of these
sensitivities remains contractive under
\Cref{cond:sym-counting-LLL}.  Thus both omitted trees and propagated
errors decay exponentially with the budget.

The deterministic algorithm truncates the recursion at budget $L$,
decreasing the budget by $|T|$ whenever it expands a $2$-tree.  
Correlation decay makes the error exponentially small in $L$.

The randomized algorithm avoids enumerating all tree terms.  It
samples one canonical $2$-tree from independent local experiments,
uses waiting-time estimators to correct rejected exploration decisions,
and uses independent recursive product estimators for the reciprocal
factors.  This gives an unbiased marginal estimator.  Separate
subcritical recurrences control its second moment and expected work.
Averaging these estimators and applying
\eqref{eq:self-reduction-technical-overview} gives the randomized
counting algorithm.

\subsubsection{Near-Linear Time Exact Sampling}

For sampling, we insert constraints one at a time while maintaining the
uniform law on the current solution space.  Thus, it suffices to
transport $X\sim\mu_{\+C\setminus\{c_0\}}$ exactly to $\mu_{\+C}$.
If $X$ satisfies $c_0$, then conditioned on this event it already has
law $\mu_{\+C}$.  Otherwise, its branch carries the unnormalized
$c_0$-violation measure
$\mu_{\+C\setminus\{c_0\}}(\,\cdot\cap\neg c_0)$, which must be
repaired to a multiple of $\mu_{\+C\setminus\{c_0\}}$ before retrying.

The $2$-tree expansion gives this repair.  Separating its alternating
terms by parity yields
\[
\text{$c_0$-violation measure}
+\text{even-$2$-tree terms}
=
\text{odd-$2$-tree terms}.
\]
We realize this identity by proposing each canonical $2$-tree with its
exact weight, or no tree; recursively deleting its telescoping
constraints and conditionally redrawing its local variables to obtain
the associated law; and using the canonical construction to transfer
the violation-and-even mixture to the odd-tree terms without changing
the assignment.  A fair coin supplies the common scale for the two
source terms.  An independent product redraw followed by the reverse
insertions restores a multiple of the original law.  Failed attempts
can therefore retry, whereas successful ones return an exact sample
from $\mu_{\+C}$.

The weighted-tree recursion is subcritical under our local lemma
condition, giving $O(1)$ expected oracle cost and
$\widetilde O(kD^2)$ expected time per insertion.  Successively
inserting all constraints gives the near-linear exact sampler.

\subsection{Related Work}

\subsubsection{Algorithmic Lov\'asz Local Lemma}

The \emph{algorithmic Lov\'asz Local Lemma} seeks efficient algorithms
for constructing a satisfying assignment of a CSP instance.
This direction was initiated by Beck~\cite{beck1991algorithmic} and
developed through a series of works~\cite{alon1991parallel,
molloy1998further,CS00,Sri08,moser2009constructive,moser2010constructive},
culminating in the seminal Moser--Tardos
algorithm~\cite{moser2010constructive}, which gives an algorithmic
version of the asymmetric LLL in the variable setting.
Subsequent works have further improved and extended this framework,
including deterministic algorithms, parallel algorithms, and
generalizations beyond the original setting~\cite{haeupler2011new,
Kolipaka2011MoserAT,chandrasekaran2013deterministic,achlioptas2016random,
kolmogorov2018commutativity,harris2019moser,harvey2020algorithmic}.

Moser--Tardos resampling restores satisfaction but need not produce a
 uniform sample. Our sampler instead maintains the
uniform solution law through dynamic rejection-and-repair updates.

\subsubsection{Approximate Counting of Random CSP Solutions}

Another related direction is the approximate counting of solutions to random CSPs.
In random CSP models, constraints are generated randomly according to a prescribed distribution, with a density parameter $\alpha>0$ for controlling the ratio between the number of constraints and variables. 
These models provide important examples of high-dimensional random constraint systems and are closely related to mean-field models in statistical physics.

Among random CSPs, random $k$-SAT is one of the most extensively studied examples in theoretical computer science. 
In this model, each constraint is generated by independently and uniformly selecting $k$ distinct variables at random and independently assigning each variable a random sign. 
Recent works have applied techniques from counting and sampling LLL to approximate counting and sampling solutions of random $k$-SAT~\cite{GGGY21,he2023improved,chen2024fast,chen2025counting}.
In particular, \cite{chen2025counting} obtained an efficient approximate counting algorithm for random $k$-SAT for densities up to a $\poly(k)$ factor below the satisfiability threshold, whose location was established through a long line of works~\cite{kirousis1998approximate,friedgut1999sharp,achlioptasAsymptoticOrderRandom2002,achlioptas2003threshold,coja2014asymptotic,ding2022satisfiability}.

\subsubsection{Version History and Concurrent Sampling Work}
The first version, \emph{A Counting Lov\'asz Local Lemma}
(\href{https://arxiv.org/abs/2608.08616v1}{arXiv:2608.08616v1},
August 2026), already contained both main counting theorems
and the $2$-tree expansion in Sections~1--5.
Our sampling work began while we were finalizing that version
and progressed from a polynomial-time approximate sampler
based on our $2$-tree expansion, through cubic- and quadratic-time
approximate samplers, to the near-linear-time exact sampler first
presented in the second version
(\href{https://arxiv.org/abs/2608.08616v2}{arXiv:2608.08616v2}).

\begingroup
\emergencystretch=1em
While we were preparing that second version,
Achlioptas~\cite{achlioptas2026sampling} and Jain, Mizgerd,
and Pham~\cite{jain2026general} posted polynomial-time
approximate samplers with parameter-dependent and
parameter-independent exponents, respectively.
Both works also extend to asymmetric criteria.
These follow-ups build on results and techniques from our
counting-only first version and are concurrent with the
sampling additions in that second version.
\par
\endgroup

The exact sampler in the second version, developed independently of
both concurrent works, used our $2$-tree expansion to achieve near-linear expected
running time for fixed local parameters under
$4\e p(D+1)^2\le1$, without invoking approximate counting.
The present third version substantially simplifies this construction, retains exact sampling under the same symmetric condition, and improves its expected oracle cost to $O(m)$.

\section{Preliminaries and Notation}\label{sec:preliminaries}

\subsection{Constraint Satisfaction Problems}
A constraint satisfaction problem (CSP) consists of a collection of local
constraints defined over a set of variables. Formally, a CSP instance, or
\emph{CSP formula}, is denoted by $\Phi=(V,\+Q,\+C)$.
Here, $V$ is a set of $n=|V|$ variables. Each variable $v\in V$ has a finite domain $Q_v$ of size $q_v\triangleq |Q_v|\geq 2$, and the assignment space is $\+Q = \bigotimes_{v\in V}Q_v$.
The set $\+C$ consists of local constraints, where each constraint $c\in\+C$ is a Boolean function $c:\bigotimes_{v\in \vbl(c)}Q_v\to\{\True,\False\}$ defined on a nonempty subset of variables $\vbl(c)\subseteq V$.
An assignment $\sigma\in\+Q$ is called \emph{satisfying} for $\Phi$ if
\[
\Phi(\sigma)\triangleq\bigwedge\limits_{c\in \+C} c\left(\sigma_{\vbl(c)}\right)=\True.
\]
We denote by $Z_\Phi$ the number of satisfying assignments of $\Phi$.
For a subset of constraints $\+E\subseteq\+C$, let $\vbl(\+E)\triangleq\bigcup_{c\in\+E}\vbl(c)$, and let $Z(\+E)\triangleq\bigl|\{\sigma\in\+Q:\sigma\text{ satisfies every }c\in\+E\}\bigr|$; in particular, $Z(\+C)=Z_\Phi$. For a subset of variables $\Lambda\subseteq V$, let $\+Q_{\Lambda}\defeq \bigotimes_{v\in \Lambda}Q_v$.  
For an assignment $\sigma\in\+Q$, we write $\sigma_\Lambda=\bigotimes_{v\in \Lambda}\sigma(v)\in \+Q_{\Lambda}$ for its restriction of $\sigma$ to $\Lambda$.  

We adopt the following simplified notation for constraint events, following
\cite{wang2024sampling}:
\begin{itemize}
    \item A constraint $c\in\+C$ also denotes the event that $c$ is satisfied;
    \item A subset of constraints $\+E\subseteq\+C$ also denotes the event that
    every constraint in $\+E$ is satisfied.
\end{itemize}

\subsection{Lov\'{a}sz Local Lemma}

In the context of the LLL, each constraint $c\in\+C$ can be viewed as a bad event $A_c$, which occurs when the assignment to $\vbl(c)$ violates the constraint $c$. 
The Lov\'asz Local Lemma provides a sufficient condition for the existence of a satisfying assignment.

\begin{theorem}[\cite{LocalLemma}]\label{thm:locallemma}
    Given a CSP formula $\Phi=(V,\+Q,\+C)$, suppose
    \begin{align}\label{llleq}
    \exists x\in (0, 1)^{\+C}\quad \text{ s.t.}\quad \forall c \in \+C:\quad
        {\Pr{\neg c}\leq x(c)\prod_{\substack{c'\in \+C\setminus \{c\}\\ \vbl(c)\cap\vbl(c')\neq \emptyset}}(1-x(c'))}.
    \end{align}
    Then  
    $$
        {\Pr{\+C}\geq \prod\limits_{c\in \+C}(1-x(c))>0}.
    $$
\end{theorem}

The LLL condition \eqref{llleq} also provides control over the effect of conditioning on all constraints being satisfied. 
The following result, due to Haeupler, Saha, and Srinivasan~\cite{haeupler2011new}, bounds the probability of an arbitrary local event under the conditional distribution.

\begin{theorem}[\text{\cite[Theorem 2.1]{haeupler2011new}}]\label{thm:HSS}
Given a CSP formula $\Phi=(V,\+Q,\+C)$ satisfying \eqref{llleq}, let $\+A$ be an event determined by the assignment on a subset of variables $\vbl(\+A)\subseteq V$. 
Then
\[
   \Pr{\+A\mid \+C}\leq \Pr{\+A}\prod_{\substack{c\in \+C\\ \vbl(c)\cap\vbl(\+A)\neq \emptyset}}(1-x(c))^{-1}.
\]
\end{theorem}

\subsection{Dependency Graph and \texorpdfstring{$2$}{2}-Trees}

The dependency graph is a fundamental structure in the Lov\'asz Local Lemma,
as it captures the interaction pattern among constraints. 
For a graph $G=(U,E)$, let $\mathrm{dist}_G(u,v)$ denote the shortest-path distance between vertices $u,v\in U$. For $u\in U$ and $S\subseteq U$, define
\[
\mathrm{dist}_G(u,S)\defeq\min_{v\in S}\mathrm{dist}_G(u,v),
\]
with the convention that $\mathrm{dist}_G(u,\emptyset)=\infty$.

\begin{definition}[Dependency Graph]\label{def:dependency-graph}
Let $(V,\+Q)$ be a variable space and let $\+C$ be a set of constraints.
\begin{itemize}
\item 
The \emph{dependency graph} $G(\+C)$ is the graph with vertex set $\+C$, where two distinct constraints $c,c'\in\+C$ are adjacent if and only if $\vbl(c)\cap\vbl(c')\neq\emptyset$.
\item 
The \emph{square graph} $G^2(\+C)$ is the graph with vertex set $\+C$, where two distinct constraints $c,c'\in\+C$ are adjacent if and only if $\mathrm{dist}_{G(\+C)}(c,c')\leq2$.
\end{itemize}
\end{definition}

For any $\+D\subseteq\+C$, we define its \emph{inclusive neighborhoods} in the dependency graph $G(\+C)$ and  the square graph $G^2(\+C)$ respectively as
\begin{align*}
\Gamma^{+}_{\+C}(\+D)
&\defeq\{c\in\+C\mid\mathrm{dist}_{G(\+C)}(c,\+D)\leq1\};\\
\Gamma^{\leq2}_{\+C}(\+D)
&\defeq\{c\in\+C\mid\mathrm{dist}_{G(\+C)}(c,\+D)\leq2\}.
\end{align*}

The notion of a $2$-tree was introduced by Alon~\cite{alon1991parallel} in the study of the algorithmic Lov\'asz Local Lemma and has subsequently played an important role in algorithmic, sampling, and counting LLL.

\begin{definition}[$2$-Tree]\label{definition:2-tree}
For a graph $G=(V,E)$, a \emph{$2$-tree} is a set $T\subseteq V$ that is independent in $G$ and connected in its square $G^2$.
\end{definition}

The following bound on connected subgraphs immediately implies an upper bound on the number of $2$-trees, since they are connected in the square graph $G^2$, which has degrees at most $\Delta^2$. 
\begin{lemma}[\cite{borgs2013left}]
  \label{lemma:number-of-component}
Let $G=(V,E)$ be a graph with maximum degree $\Delta$, and let $v\in V$.
Then the number of connected induced subgraphs of size $\ell$ containing $v$ is at most   $(\mathrm{e}\Delta)^{\ell - 1}$.
\end{lemma}
\begin{corollary}
  \label{cor:num-2-tree}
Let $G=(V,E)$ be a graph with maximum degree $\Delta$, and let $v\in V$.
Then the number of $2$-trees of size $\ell$ containing $v$ is at most
  $(\mathrm{e}\Delta^2)^{\ell - 1}$.
\end{corollary}

\subsection{Distributions and Sampling Conventions}
\label{sec:sampling-conventions}
For a constraint family $\CF\subseteq\+C$, let
\[
    \Sol_\CF=\{x\in\+Q:c(x)=\True\text{ for every }c\in\CF\}.
\]
Assignments in $\Sol_\CF$ are defined on all of $V$, including
variables that occur in no constraint of $\CF$, so
$Z(\CF)=|\Sol_\CF|$. A \emph{local assignment} on $W\subseteq V$
assigns a value to every variable in $W$.

Throughout, $\mathbb P$ is the uniform product distribution on $\+Q$.
For $\CF\subseteq\+C$, let
\[
 \mu_\CF=\mathbb P(\,\cdot\mid\CF),
\]
which is uniform on $\Sol_\CF$.  The local lemma guarantees
$\mathbb P(\CF)>0$, and $\mu_\varnothing=\mathbb P$.  The sampling
algorithm uses independent uniform variable draws and calls to the
evaluation oracle; auxiliary coins are not charged as oracle calls.
The symbol $\pr{\cdot}$ denotes probability in the experiment under
discussion, whereas $\Pr{\cdot}$ is reserved for probability in the product
space.

\subsection{Computational Model}
\label{sec:computational-model}
We use a random-access model with binary-encoded integers and rationals,
charging arithmetic and data access by their bit costs.
Evaluation-oracle calls have unit cost.
Domains are indexed, and constraint scopes and variable-incidence
lists are supplied explicitly.
Running time includes preprocessing, preparing oracle queries, and
exact random choices generated from unbiased bits.

\section{Two-Tree Expansion}\label{sec:algo}
\subsection{Constraint Marginal Probability}
Let $\Phi=(V,\+Q,\+C)$ be a CSP formula satisfying
\Cref{cond:sym-counting-LLL}. By \Cref{thm:locallemma},
$Z(\+D)>0$ for every $\+D\subseteq\+C$.
Fix any $c_0\in\+C$.
A central quantity estimated by our algorithm is the following marginal probability of a constraint:
\begin{equation}
\label{eq:def-marginal-violation}
    r_{\+C,c_0}\defeq 1-\frac{Z(\+C)}{Z(\+C\setminus \{c_0\})}=\Pr{\neg c_0\mid\+C\setminus\{c_0\}}=\frac{\Pr{\neg c_0\land (\+C\setminus\{c_0\})}}{\Pr{\+C\setminus\{c_0\}}}.
\end{equation}

The following simple upper bound holds for $r_{\+{C},c_0}$.

\begin{lemma}\label{lem:marginal-bound}
Suppose that $\Phi=(V,\+Q,\+C)$ satisfies \Cref{cond:sym-counting-LLL}. Then, for every $c_0\in \+C$,
\[
0\leq r_{\+C,c_0}\leq 1.2p.
\]
\end{lemma}
\begin{proof}
If $p=0$, every constraint is always satisfied and the claim is
immediate. Assume $p>0$.

By \Cref{cond:sym-counting-LLL}, $4\e p(D+1)^2\leq1$, so $1.2pD\leq0.1$, and hence
$(1-1.2p)^D\geq1-1.2pD\geq0.9$.
Consequently, $p\leq1.2p(1-1.2p)^D$.
Every constraint in $\+C\setminus\{c_0\}$ has at most $D$ neighbors,
so $x(c)=1.2p$ satisfies \eqref{llleq} for this subformula.
Applying \Cref{thm:HSS} to $\+C\setminus\{c_0\}$ gives
$r_{\+C,c_0}\leq p(1-1.2p)^{-D}\leq1.2p$.
\end{proof}

By the principle of inclusion-exclusion, we have
\[
\Pr{\neg c_0\land (\+C\setminus\{c_0\})}
=\sum\limits_{\substack{c_0\in\+D\subseteq \+C }}(-1)^{\abs{\+D}-1}\Pr{\bigwedge\limits_{c\in \+D } \neg c}.
\]
Consequently, substituting this expansion into the definition of $r_{\+C,c_0}$ gives
\begin{equation}\label{eq:inclusion-exclusion}
    r_{\+C,c_0}=\frac{1}{\Pr{\+C\setminus\{c_0\}}}\sum\limits_{\substack{c_0\in\+D\subseteq \+C }}(-1)^{\abs{\+D}-1}\Pr{\bigwedge\limits_{c\in \+D } \neg c}.
\end{equation}

\subsection{Canonical \texorpdfstring{$2$}{2}-Tree Representation}
The inclusion-exclusion expansion \eqref{eq:inclusion-exclusion} involves exponentially many constraint subsets $c_0\in\+D\subseteq\+C$.
To approximate the marginal probability $r_{\+C,c_0}$, we focus on subsets $\+D$ with non-negligible contributions:
\[
p(\+D)=\Pr{\bigwedge\limits_{c\in \+D } \neg c}.
\]
The probability $p(\+D)$ is small whenever $\+D$ contains a large independent set in the dependency graph $G(\+C)$.
This motivates extracting a ``canonical'' $2$-tree $T\subseteq\+D$ from each constraint subset $\+D$.
A  large canonical $2$-tree serves as a certificate that the probability $p(\+D)$ is exponentially small.

Fix a constraint set $\+C$, an arbitrary ordering of its constraints, and a root constraint $c_0\in\+C$.
For any constraint subset $\+D\subseteq\+C$ containing $c_0$, the procedure 
$$\tr(\+D)=\tr_{\+C,c_0}(\+D)$$ 
in \Cref{Alg:gen-tree} constructs a \emph{canonical $2$-tree} that contains $c_0$ and is contained in $\+D$.
The construction is greedy: at each step, it selects the lowest-index constraint at graph distance exactly $2$ from the current $2$-tree, and accepts it if it belongs to $\+D$; otherwise, it records the constraint as rejected.

\begin{algorithm}[ht]
\caption{$\tr_{\+C,c_0}(\+D)$: Construction of a Canonical $2$-Tree } \label{Alg:gen-tree}
\SetKwInOut{Input}{Input}
\SetKwInOut{Output}{Output}
\SetKwIF{WP}{ElseIf}{Else}{with probability}{do}{else if}{else}{endif}
\Input{A constraint set $\+D\subseteq\+C$ containing $c_0$.} 
\Output{A $2$-tree $T$ in $G(\+C)$ containing $c_0$ and satisfying $T\subseteq\+D$.}
Initialize $R\gets \emptyset$, $T\gets \{c_0\}$\;
\While{$\exists c\in \+C\setminus R$ such that $\mathrm{dist}_{G(\+C)}(c,T)= 2$\label{Line:chosen-condition}}
{
    Let $c$ be the lowest-index constraint satisfying the above condition\label{Line:gen-tree-choose}\;
    \If{$c\in \+D$}{
        $T\gets T\cup \{c\}$\tcp*{accepted set}
    }\Else{
        $R\gets R\cup \{c\}$\tcp*{rejected set}\label{Line:gen-tree-reject}
    }
}
\Return $T$\;
\end{algorithm}

For any constraint set $\+C$ and $c\in\+C$, let
$\mathfrak{T}_{\+C,c}$ denote the family of all $2$-trees in $G(\+C)$
containing~$c$.

The following lemma shows that $\tr(\+D)=\tr_{\+C,c_0}(\+D)$ maps every constraint subset $\+D\subseteq\+C$ containing $c_0$ to a $2$-tree in $\mathfrak{T}_{\+C,c_0}$, and that every $2$-tree in $\mathfrak{T}_{\+C,c_0}$ is a fixed point of this mapping.

\begin{lemma}
    \label{lem:property-2-tree}
    For any $c_0 \in \+D\subseteq \+C$,  $\tr(\+D)\in \mathfrak{T}_{\+C,c_0}$.
    Moreover, for any $T\in \mathfrak{T}_{\+C,c_0}$, $\tr(T)=T$.
\end{lemma}

\begin{proof}
Let $T=\tr(\+D)$. By construction, $c_0\in T\subseteq\+D$.
Initially, $T=\{c_0\}$.
Every added constraint $c$ satisfies $\dist_{G(\+C)}(c,T)=2$:
it is nonadjacent to every constraint in $T$ but adjacent in
$G^2(\+C)$ to some constraint in $T$.
Thus $T$ remains independent in $G(\+C)$ and connected in
$G^2(\+C)$, so the final output is a $2$-tree containing $c_0$.

Now fix $T\in\mathfrak{T}_{\+C,c_0}$ and let $T'\subseteq T$
be the final accepted set on input $T$.
If $T'\subsetneq T$, connectivity of $G^2(\+C)[T]$ gives
$u\in T'$ and $v\in T\setminus T'$ with
$\dist_{G(\+C)}(u,v)=2$.
Since $T$ is independent in $G(\+C)$,
$\dist_{G(\+C)}(v,T')=2$.
Moreover, $v$ cannot have been rejected because it belongs to
the input set $T$.
Thus $v$ still satisfies the while-condition at termination,
a contradiction. Hence $\tr(T)=T$.
\end{proof}

The following notion of free constraints is used to formulate our 2-tree expansion.

\begin{definition}[Free Constraints]\label{def:bad-and-free}
For each $T\in\mathfrak{T}_{\+C,c_0}$, let $R=R(T)$ denote the final rejected set constructed during the execution of $\tr(T)$. The set of \emph{free constraints} with respect to $T$ is defined as
\[
    \fr(T)=\fr_{\+C,c_0}(T)\defeq \+C\setminus (T\cup R).
\]

\end{definition}

The terminology of \emph{free constraints} comes from the following inverse question:

  Given a canonical $2$-tree $T=\tr(\+D)$, which constraint sets $\+D$ produce the same canonical $2$-tree?  

The membership of some constraints in $\+D$ is fixed: constraints in $T$ must belong to $\+D$, while constraints in the rejected set $R$ must not belong to $\+D$. 
The remaining constraints in $\fr(T)=\+C\setminus(T\cup R)$ are free, since adding or removing any subset of them does not change the resulting canonical $2$-tree. 

Indeed, the free constraints capture exactly the degrees of freedom in the preimage $\tr^{-1}(T)$, as formalized by the following lemma.

\begin{lemma}[Canonical 2-tree inversion]\label{lem:property-gen-tree}
    Let $\+C$ be a constraint set and let $c_0\in \+C$. For every $T \in \mathfrak{T}_{\+C,c_0}$, 
    $$\{\+D \subseteq \+C \mid c_0 \in \+D, \; \tr(\+D) = T\} = \{S \cup T \mid S \subseteq \fr(T)\}.$$
\end{lemma}

\begin{proof}
By \Cref{lem:property-2-tree}, $\tr(T)=T$ for every $T\in\mathfrak{T}_{\+C,c_0}$.

We first show that every $T\cup S$ with $S\subseteq\fr(T)$ satisfies $\tr(T\cup S)=T$. 
Compare the executions on input sets $T$ and $T\cup S$.
We claim that they have identical accepted and rejected sets throughout the execution. 
Suppose this holds up to some iteration, and let $d$ be the next selected constraint. 
If $d\in T$, both executions accept $d$. 
Otherwise, the execution on input set $T$ rejects $d$, so $d\in R$. Since $S\subseteq\fr(T)=\+C\setminus(T\cup R)$, we have $d\notin T\cup S$, and hence the second execution also rejects $d$.
Therefore the two executions are identical and $\tr(T\cup S)=T$.

Conversely, suppose that $\tr(\+D)=T$. Since every accepted constraint must belong to the input set, we have $T\subseteq\+D$. Compare the execution on input set $\+D$ with the execution on input set $T$. 
At any common state, let $d$ be the next selected constraint. 
If $d\in T$, then both executions accept $d$. 
Otherwise, if $d\in\+D$, the execution on input set $\+D$ would accept $d$, contradicting that its final output is $T$. 
Hence $d\notin\+D$, and both executions reject $d$. 
Thus the two executions are identical, and every constraint rejected in the execution on input set $T$ is also rejected in the execution on input set $\+D$. 
In particular, $\+D\cap R=\emptyset$.
Together with $T\subseteq\+D$, this implies $\+D=T\cup S$ for some $S\subseteq\+C\setminus(T\cup R)=\fr(T)$.
\end{proof}

\subsection{\texorpdfstring{$2$}{2}-Tree Expansion}
The following $2$-tree expansion gives a recursion for constraint marginals.

\begin{lemma}[Recursion through $2$-tree expansion]\label{lemma:2-tree-expansion-recursion}

For each $T\in \mathfrak{T}_{\+C,c_0}$, enumerate the constraints in
\[
\+C\setminus(\{c_0\}\cup \fr(T))
=
\set{c^T_1,\ldots,c^T_{\ell^T}},
\]
where $\ell^{T}\defeq\abs{\+C\setminus(\{c_0\}\cup \fr(T))}$.
For each $0\le i\le\ell^{T}$, define 
\[
\+C_i^{T} \defeq \fr(T)\cup \set{c^T_1,\dots,c^T_i}.
\]
Then, the marginal probability $r_{\+C,c_0}$ satisfies
\begin{equation}\label{eq:2-tree-expansion-recurrence}
r_{\+C,c_0}=\sum\limits_{T\in \mathfrak{T}_{\+C,c_0} }(-1)^{\abs{T}-1} 
    \Pr{
        \bigwedge_{c\in T}\neg c
        \,\middle|\,
        \fr(T)
    }
\cdot \prod\limits_{i=1}^{\ell^{T}}\frac{1}{1-r_{\+C^{T}_i,c^{T}_i}}.
\end{equation}
\end{lemma}
\begin{proof}
We first establish the expansion pointwise. For every $B\subseteq\+C$,
\begin{equation}\label{eq:2-tree-pointwise}
\sum_{T\in\mathfrak{T}_{\+C,c_0}}(-1)^{|T|-1}
\one{T\subseteq B\subseteq T\cup R(T)}
=\one{B=\{c_0\}}.
\end{equation}
If $c_0\notin B$, both sides are zero. Otherwise, regroup
$\sum_{c_0\in\+D\subseteq B}(-1)^{|\+D|-1}$ according to $\tr(\+D)$.
By \Cref{lem:property-gen-tree}, the contribution of each $T\subseteq B$ is
\[
(-1)^{|T|-1}\sum_{S\subseteq B\cap\fr(T)}(-1)^{|S|},
\]
which equals $(-1)^{|T|-1}$ when $B\cap\fr(T)=\varnothing$ and zero otherwise.
Since $\fr(T)=\+C\setminus(T\cup R(T))$, this proves
\eqref{eq:2-tree-pointwise}: the ungrouped sum equals
$\one{B=\{c_0\}}$.

Let $X$ be a uniformly random assignment, and let $B$ be the set of constraints violated by $X$.
Taking expectations in \eqref{eq:2-tree-pointwise} and dividing by
$\Pr{\+C\setminus\{c_0\}}$ gives
\begin{equation}\label{eq:2-tree-expansion-refined}
\begin{aligned}
r_{\+C,c_0}
&=\frac{1}{\Pr{\+C\setminus\{c_0\}}}
\sum_{T\in\mathfrak{T}_{\+C,c_0}}(-1)^{|T|-1}
\Pr{\fr(T)\land\tp{\bigwedge_{c\in T}\neg c}}\\
&=\sum_{T\in\mathfrak{T}_{\+C,c_0}}(-1)^{|T|-1}
\Pr{\tp{\bigwedge_{c\in T}\neg c}\,\middle|\,\fr(T)}
\cdot\frac{\Pr{\fr(T)}}{\Pr{\+C\setminus\{c_0\}}}.
\end{aligned}
\end{equation}
It remains to express the ratio
$\Pr{\fr(T)}/\Pr{\+C\setminus\{c_0\}}$ recursively.
By the definition of $\+C_i^T$ and a telescoping product,
\[
\frac{\Pr{\fr(T)}}{\Pr{\+C\setminus\{c_0\}}}
=\frac{Z(\+C_0^T)}{Z(\+C\setminus\{c_0\})}
=\prod_{i=1}^{\ell^T}\frac{Z(\+C_{i-1}^T)}{Z(\+C_i^T)}
=\prod_{i=1}^{\ell^T}\frac{1}{1-r_{\+C_i^T,c_i^T}}.
\]
Substituting this identity into \eqref{eq:2-tree-expansion-refined}
gives the recursion \eqref{eq:2-tree-expansion-recurrence}.
\end{proof}

\subsection{Correlation Decay for the \texorpdfstring{$2$}{2}-Tree Expansion}
We establish a correlation decay property for the $2$-tree expansion. 
Define the following function corresponding to the recursion in \eqref{eq:2-tree-expansion-recurrence}:
\begin{align}
    f(\bm{x})\defeq \sum\limits_{T\in \mathfrak{T}_{\+C,c_0}}(-1)^{\abs{T}-1}\cdot
        \Pr{
        \bigwedge_{c\in T}\neg c
        \,\middle|\,
        \fr(T)
    }
\cdot \prod\limits_{i=1}^{\ell^{T}}\frac{1}{1-x_{\+C^{T}_i,c^{T}_i}}.
\end{align}

\begin{lemma}[Correlation decay for the $2$-tree expansion]\label{lem:2-tree-expansion-correlation-decay}
    Suppose $\Phi=(V,\+Q,\+C)$ satisfies \Cref{cond:sym-counting-LLL}. 
    Let $\alpha=2.1(D+1)^2$ and ${\beta=  p(1-2p)^{-(D+1)^2}}$.
    Then, for every $T\in \mathfrak{T}_{\+C,c_0}$ {and any vector $\bm{x}=(x_{\+C^{T}_i,c^{T}_i})_{T\in \mathfrak{T}_{\+C,c_0},1\leq i\leq \ell^{T}}$ }  satisfying 
    ${-2p\leq x_{\+C^{T}_i,c^{T}_i}\leq 2p}$,
    \begin{equation}\label{eq:2-tree-expansion-single-bound}
        \Pr{
        \bigwedge_{c\in T}\neg c
        \,\middle|\,
        \fr(T)
    }
\prod\limits_{i=1}^{\ell^{T}}\frac{1}{1-{x_{\+C^{T}_i,c^{T}_i}}}\leq \beta^{|T|}.
    \end{equation}
    Moreover, for any two vectors
    $\bm{x}=(x_{\+C^{T}_i,c^{T}_i})_{T\in \mathfrak{T}_{\+C,c_0},1\leq i\leq \ell^{T}}$ 
    and 
    $\bm{y}=(y_{\+C^{T}_i,c^{T}_i})_{T\in \mathfrak{T}_{\+C,c_0},1\leq i\leq \ell^T}$
    satisfying 
    ${-2p\leq x_{\+C^{T}_i,c^{T}_i},y_{\+C^{T}_i,c^{T}_i}\leq 2p}$
    for every  $T\in \mathfrak{T}_{\+C,c_0}$ and  $1\leq i\leq \ell^T$,
    we have
    \begin{equation}\label{eq:2-tree-expansion-total-bound}
    \abs{f(\bm{x})-f(\bm{y})}\leq \alpha \sum\limits_{T\in \mathfrak{T}_{\+C,c_0}}{(1.2\beta)}^{\abs{T}} \max\limits_{1\leq i\leq \ell^{T}}\abs{x_{\+C^{T}_i,c^{T}_i}-y_{\+C^{T}_i,c^{T}_i}}.
    \end{equation}
    The maximum over an empty index set is interpreted as zero.
\end{lemma}

\begin{proof}
If $p=0$, both bounds are immediate; assume $p>0$.
Fix $T\in\mathfrak T_{\+C,c_0}$. Since
$p\leq1.2p(1-1.2p)^D$, $\fr(T)$ satisfies \eqref{llleq}
with $x(c)=1.2p$. Applying \Cref{thm:HSS} to $\fr(T)$ and
$\bigwedge_{c\in T}\neg c$, whose scope meets at most $D|T|$
constraints in $\fr(T)$, gives
\begin{equation}
\label{eq:conditional-probability-upper-bound}
    \Pr{
        \bigwedge_{c\in T}\neg c
        \,\middle|\,
        \fr(T)
    }
    \leq
    \Pr{
        \bigwedge_{c\in T}\neg c
    }
    (1-1.2p)^{-D|T|}
    \leq
    p^{|T|}(1-1.2p)^{-D|T|}.
\end{equation}
We next bound the value of $\ell^T$.
Every constraint is added to $R$ in \Cref{Alg:gen-tree}
only when it has distance exactly $2$ from the current $2$-tree.
Hence, the number of rejected constraints in $R$ is at most $D^2|T|$.
Consequently,
    \begin{equation}\label{eq:ell-T-upper-bound}
    \ell^{T}=\abs{\+C\setminus(\{c_0\}\cup \fr(T))}=\abs{R\cup T}-1\leq D^2|T|+|T|-1.
\end{equation}
{Note that $-2p\leq x_{\+C_i^T,c_i^T}\leq 2p$} for every $1\le i\le \ell^T$.
It then follows from \eqref{eq:conditional-probability-upper-bound} and \eqref{eq:ell-T-upper-bound} that
\[
\Pr{
    \bigwedge_{c\in T}\neg c
    \,\middle|\,
    \fr(T)
}
\prod_{i=1}^{\ell^T}
\frac{1}{1-{x_{\+C_i^T,c_i^T}}}
\leq
p^{|T|}
(1-1.2p)^{-D|T|}
(1-{2p})^{-\ell^T}
\leq
\beta^{|T|}.
\]
We now prove the bound in \eqref{eq:2-tree-expansion-total-bound}.
The following claim bounds the sensitivity of the product term.
\begin{claim} \label{claim:product-lipshitz}
For any $0<t<1$, $\ell\geq1$, and any two vectors
$(x_i)_{i\in[\ell]},(y_i)_{i\in[\ell]}\in[{-t},t]^\ell$,
\begin{equation}\label{eq:product-lipschitz-polymer}
\abs{
\prod_{i=1}^{\ell}\frac{1}{1-x_i}
-
\prod_{i=1}^{\ell}\frac{1}{1-y_i}}
\leq
\frac{\ell}{(1-t)^{\ell+1}} \cdot 
\max_{1\leq i\leq\ell}|x_i-y_i|.
\end{equation}
\end{claim}

Applying the claim with $t={2p}$, together with
${-2p}\leq x_{\+C_i^T,c_i^T},y_{\+C_i^T,c_i^T}\leq 2p$, gives
\begin{align*}
|f(\bm{x})-f(\bm{y})|
\leq&
\sum_{T\in\mathfrak{T}_{\+C,c_0}}
\Pr{
    \bigwedge_{c\in T}\neg c
    \,\middle|\,
    \fr(T)
}
\cdot
\left|
\prod_{i=1}^{\ell^T}
\frac{1}{1-x_{\+C_i^T,c_i^T}}
-
\prod_{i=1}^{\ell^T}
\frac{1}{1-y_{\+C_i^T,c_i^T}}
\right|
\\
\leq&
\sum_{T\in\mathfrak{T}_{\+C,c_0}}
\Pr{
    \bigwedge_{c\in T}\neg c
    \,\middle|\,
    \fr(T)
}
\frac{\ell^T}{(1-{2p})^{\ell^T+1}}
\cdot
\max_{1\leq i\leq\ell^T}
|x_{\+C_i^T,c_i^T}-y_{\+C_i^T,c_i^T}|.
\end{align*}
Using \eqref{eq:conditional-probability-upper-bound} and
\eqref{eq:ell-T-upper-bound}, we obtain
\[
|f(\bm{x})-f(\bm{y})|
\leq
\alpha
\sum_{T\in\mathfrak{T}_{\+C,c_0}}
{(1.2\beta)}^{|T|}
\max_{1\leq i\leq\ell^T}
\abs{x_{\+C_i^T,c_i^T}-y_{\+C_i^T,c_i^T}}.
\]
Here we used the elementary inequality
$x\leq2.1(1.2)^x$ for $x\geq1$ to absorb the polynomial factor
arising from $\ell^T$ into the exponential term ${(1.2\beta)^{|T|}}$.
This proves \eqref{eq:2-tree-expansion-total-bound}.
\end{proof}

\begin{proof}[Proof of \Cref{claim:product-lipshitz}]
We use the following telescoping decomposition:
\begin{equation}\label{eq:lipshitz-decomposition-1}
\prod_{i=1}^{\ell}\frac{1}{1-x_i}-\prod_{i=1}^{\ell}\frac{1}{1-y_i}=\sum\limits_{j=1}^{\ell}\tp{\prod\limits_{i<j}\frac{1}{1-x_i}}\tp{\frac{1}{1-x_j}-\frac{1}{1-y_j}}\tp{\prod\limits_{i>j}\frac{1}{1-y_i}}.
\end{equation}
Fix $0\leq t<1$ and assume that $(x_i)_{i\in[\ell]},(y_i)_{i\in[\ell]}\in[{-t},t]^\ell$. For every $1\leq j\leq\ell$,
\begin{equation}\label{eq:lipshitz-decomposition-2}
\abs{\frac{1}{1-x_j}-\frac{1}{1-y_j}}=\frac{\abs{x_j-y_j}}{(1-x_j)(1-y_j)}\leq \frac{\abs{x_j-y_j}}{(1-t)^2}.
\end{equation}
Combining \eqref{eq:lipshitz-decomposition-1} and \eqref{eq:lipshitz-decomposition-2}, and using the fact that every remaining reciprocal factor is at most $(1-t)^{-1}$, we obtain
\[
\abs{
\prod_{i=1}^{\ell}\frac{1}{1-x_i}
-
\prod_{i=1}^{\ell}\frac{1}{1-y_i}}\leq \sum\limits_{i=1}^{\ell}\frac{\abs{x_i-y_i}}{(1-t)^{\ell+1}}\leq \frac{\ell}{(1-t)^{\ell+1}}\max\limits_{1\leq i\leq \ell}\abs{x_i-y_i}.
\]
This proves the claim.
\end{proof}

\section{Deterministic Approximate Counting}\label{sec:deterministic-counting}

\subsection{Deterministic Marginal Estimation}

We estimate $r_{\+C,c_0}$ deterministically by truncating the recursive computation in \eqref{eq:2-tree-expansion-recurrence}.

Let $L\geq 0$ be an integer truncation budget. 
We define recursively an estimate $\hat{r}^{(L)}_{\+C,c_0}$ of $r_{\+C,c_0}$. 
For the base case $L=0$, we set
\[
    \hat{r}^{(0)}_{\+C,c_0}=0.
\]
For $L>0$, recall the notation $\ell^T$, $\{c_i^T\}_{1\leq i\leq\ell^T}$, and $\{\+C_i^T\}_{0\leq i\leq\ell^T}$ from \Cref{lemma:2-tree-expansion-recursion}. 
We recursively define
\begin{equation}
\label{eq:2-tree-recursion-truncation}
\hat{r}^{(L)}_{\+C,c_0}
=
\sum_{\substack{T\in \mathfrak{T}_{\+C,c_0}\\ |T|\leq L}}
(-1)^{|T|-1}
\Pr{
    \bigwedge_{c\in T}\neg c
    \,\middle|\,
    \fr(T)
}\cdot
\prod_{i=1}^{\ell^T}
\frac{1}{
    1-\hat{r}^{(L-|T|)}_{\+C_i^T,c_i^T}
}
.
\end{equation}

Each step expands only $2$-trees $T$ with $|T|\le L$ and passes budget $L-|T|$ to the corresponding recursive calls.

We first bound the cost of computing $\hat{r}^{(L)}_{\+C,c_0}$, in terms of both running time and oracle costs.

\begin{lemma}[Efficiency of the marginal estimator]\label{lem:ratio-2-tree-efficiency}
Let $\Phi=(V,\+Q,\+C)$ satisfy \Cref{cond:sym-counting-LLL}.
After shared preprocessing in $\widetilde O(n+kmD)$ time,
for any $c_0\in\+C$ and $L\ge0$, the estimate
$\hat r^{(L)}_{\+C,c_0}$ can be computed using $q^{O(kDL)}$
evaluation-oracle calls within $\widetilde O(q^{O(kDL)})$ time.
\end{lemma}

It can be verified by induction that the recursion tree of $\hat{r}^{(L)}_{\+C,c_0}$ has size $(D+1)^{O(L)}$.
The main computational bottleneck lies in evaluating conditional probabilities of the form
$\Pr{
        \bigwedge_{c\in T}\neg c
        \,\middle|\,
        \fr(T)
    }$.
A direct computation is inefficient because the constraint set $\fr(T)$ may be large even when $T$ is small. 
The following lemma shows that this conditional probability can be evaluated efficiently.

\begin{lemma}\label{lem:free-1-probability}
Fix $T\in \mathfrak{T}_{\+C,c_0}$ and define
\[
\fr^{(1)}(T)
\defeq
\{c\in\fr(T):\dist_{G(\+C)}(c,T)=1\}.
\]
Then $|\fr^{(1)}(T)|\leq (D-1)|T|+1$ and $\fr^{(1)}(T)$ can be constructed using $\poly(k,q,D,|T|)$ graph operations.
Moreover,
\[
\Pr{
    \bigwedge_{c\in T}\neg c
    \,\middle|\,
    \fr(T)
}
=
\Pr{
    \bigwedge_{c\in T}\neg c
    \,\middle|\,
    \fr^{(1)}(T)
}.
\]
\end{lemma}
\begin{proof}
Let $N(T)$ be the set of constraints adjacent to $T$, and let
$d_T(b)$ be the number of constraints in $T$ adjacent to $b\in N(T)$.
Since $T$ is independent in $G(\+C)$ and connected in $G^2(\+C)$,
the sets of $T$-neighbors of the constraints in $N(T)$ connect $T$.
Adding such a set can reduce the number of components by at most
$d_T(b)-1$, so
$\sum_{b\in N(T)}(d_T(b)-1)\ge |T|-1$. As
$\sum_{b\in N(T)}d_T(b)\le D|T|$ and
$\fr^{(1)}(T)\subseteq N(T)$, we obtain
\[
|\fr^{(1)}(T)| \leq |N(T)| \leq D|T|-(|T|-1) = (D-1)|T|+1.
\]
The set $\fr^{(1)}(T)$ can be constructed by enumerating the constraints adjacent to $T$ in $G(\+C)$, which uses $\poly(k,q,D,|T|)$ graph operations.

It remains to prove the conditional probability identity. When
\Cref{Alg:gen-tree} terminates on input $T$, every constraint at final
distance $2$ from $T$ has been rejected. Thus no constraint in
$\fr(T)$ has distance $2$ from $T$.
Hence, every constraint in $\fr(T)\setminus\fr^{(1)}(T)$ is at distance at least $3$ from $T$. 
Therefore, in $G(\+C)$ there is no edge between
$T\cup\fr^{(1)}(T)$ and $\fr(T)\setminus\fr^{(1)}(T)$.
The corresponding events depend on disjoint variable sets, so
conditioning on $\fr(T)\setminus\fr^{(1)}(T)$ does not change the
conditional probability of the violations in $T$ given
$\fr^{(1)}(T)$. This proves the claimed identity.
\end{proof}

We are now ready to prove \Cref{lem:ratio-2-tree-efficiency}.
\begin{proof}[Proof of \Cref{lem:ratio-2-tree-efficiency}]
Dependency lists are shared, and recursive subfamilies are represented by
an array recording membership in the current family. Each recursive call
changes only the necessary entries and restores them on return.

By \Cref{cor:num-2-tree}, the number of $2$-trees of size $t$ containing
$c_0$ is at most $(\mathrm eD^2)^{t-1}$.
Apply the connected induced-subgraph enumeration of Patel and Regts~\cite[Lemma~3.4]{patel2017deterministic} to $G^2(\+C)$, rooted at $c_0$, and retain only sets independent in $G(\+C)$. 
Thus all such $2$-trees can be enumerated using $\poly(k,q,D,t)\cdot(\e D^2)^{t-1}$ operations.

Fix a $2$-tree $T$. By \Cref{lem:free-1-probability}, the set
$\fr^{(1)}(T)$ can be constructed using $\poly(k,q,D,|T|)$ graph
operations. The same lemma reduces the conditional probability in the
recursion to conditioning on $\fr^{(1)}(T)$, so it can be evaluated by
enumerating assignments on the variables appearing in
$T\cup\fr^{(1)}(T)$.
Since $|T\cup\fr^{(1)}(T)|\leq D|T|+1$,
the number of variables involved is at most $k(D|T|+1)$.
Therefore the cost of evaluating this conditional probability is $\poly(k,q,D,|T|) \cdot q^{k(D|T|+1)}$.

For a fixed recursion depth $L$, let $\tau(L)$ denote the maximum number of operations for computing $\hat r^{(L)}_{\+D,c}$ over all $\+D\subseteq\+C$ and $c\in\+D$.
For every enumerated $T$, the recursive calls decrease the truncation depth by $|T|$. 
Also, the number of recursive subinstances generated by $T$ is at most
$|\Gamma_{\+C}^{\le2}(T)|\leq (D+1)^2|T|$.
Therefore,
    $$\tau(L) \leq \sum_{t = 1}^L \poly(k,q,D,t) \cdot (\e D^2)^{t - 1} \cdot \left( q^{k(Dt + 1)} + (D+1)^2 t \cdot \tau(L - t) \right).$$
    An induction on $L$ shows that for any $L \ge 1$,
    \[
    \tau(L)
    \leq
    \left(
    \poly(k, q, D) \cdot q^{k(D + 1)}
    \right)^L \le q^{O(kDL)}.
    \]
For $L\ge1$, the arithmetic formula has size $S=q^{O(kDL)}$
and leaves of bit length $O(kDL\log q)$.
Rational operations add operand lengths up to a constant,
so intermediate lengths are $O(S(1+kDL\log q))$.
Polynomial arithmetic overhead and bit-cost indexing therefore
preserve the claimed $\widetilde O(q^{O(kDL)})$ bound.
Denominators are nonzero by the local lemma and
\Cref{lem:ratio-2-tree-correctness}.
\end{proof}
We next bound the error of the estimate $\hat{r}^{(L)}_{\+C,c_0}$. { Note that we simultaneously show the stability of the intermediate marginal estimates $\hat{r}^{(L)}_{\+C,c_0}$ defined recursively in \eqref{eq:2-tree-recursion-truncation}.  }

\begin{lemma}[Accuracy of the marginal estimator]
\label{lem:ratio-2-tree-correctness}
Let $\Phi=(V,\+Q,\+C)$ satisfy \Cref{cond:sym-counting-LLL}.
For any constraint $c_0\in\+C$ and any truncation budget $L\geq0$, 
\[
    { -2p\leq \hat{r}^{(L)}_{\+C,c_0}\leq 2p,\quad }\abs{\hat{r}^{(L)}_{\+C,c_0}-r_{\+C,c_0}}
    \leq
    4\cdot(0.9)^L .
\]
\end{lemma}

\begin{proof}
If $p=0$, all marginal probabilities and their truncated
estimates are zero, so the claim is immediate.
Hence assume $p>0$.
Every subformula $(V,\+Q,\+D)$ with $\+D\subseteq\+C$ satisfies
\Cref{cond:sym-counting-LLL}. Hence \Cref{lem:marginal-bound} gives
$r_{\+D,c}\in[0,1.2p]$ for every $c\in\+D$.

Define
\[
d_L
\defeq
\max_{\substack{\+D\subseteq\+C\\ c\in\+D}}
\abs{\hat{r}^{(L)}_{\+D,c}-r_{\+D,c}} .
\]
{We prove by induction for every $L\geq0$ that $d_L\leq 4 \cdot (0.9)^L$  and $-2p\leq \hat{r}^{(L)}_{\+D,c}\leq 2p$ for every $\+D\subseteq\+C$ and $c\in\+D$, which implies the lemma.}

The base case $L=0$ is immediate, since $\hat{r}^{(0)}_{\+D,c}=0$ for each $\+D\subseteq \+C$ and $r_{\+D,c}\leq1.2p<4$.

For $L\ge1$, every retained tree satisfies $1\leq |T|\leq L$,
so its recursive calls have smaller budget $L-|T|$.
The induction hypothesis therefore bounds each denominator
below by $1-2p>0$.
By the first bound in
\Cref{lem:2-tree-expansion-correlation-decay},
\[
\Pr{
    \bigwedge_{c\in T}\neg c
    \,\middle|\,
    \fr(T)
}
\prod_{i=1}^{\ell^T}
\frac{1}{1-\hat r^{(L-|T|)}_{\+C_i^T,c_i^T}}
\leq \beta^{|T|},
\qquad
\beta=p(1-2p)^{-(D+1)^2}.
\]
Summing over the retained trees and applying \Cref{cor:num-2-tree}
gives
\[
\abs{\hat r^{(L)}_{\+C,c_0}}
\leq\sum_{t=1}^{L}(\e D^2)^{t-1}\beta^t.
\]

The proof of \Cref{lem:2-tree-expansion-correlation-decay} applies
unchanged to every subfamily of rooted $2$-trees. Therefore, using its first
bound for the omitted terms with $|T|>L$, its second bound for the retained
terms with $|T|\leq L$, and applying \Cref{cor:num-2-tree}, we obtain
\[
d_L
\leq
\sum_{t=L+1}^{\infty}(\e D^2)^{t-1}\beta^t
+
{1.2\alpha\beta d_{L-1}}
+
{1.2\alpha\beta}
\sum_{t=2}^{L}
({1.2\e D^2\beta})^{t-1}d_{L-t},
\]
where $\alpha=2.1(D+1)^2$.

We next verify the constants. Under
\Cref{cond:sym-counting-LLL}, we have
\[
{\e D^2\beta<0.32},
\qquad
{\alpha\beta<0.25},
\qquad
{\beta<1.22p }.
\]
Consequently, by the induction hypothesis,
\[
{\abs{\hat{r}^{(L)}_{\+C,c_0}}\leq \sum\limits_{t=1}^{L}(0.32)^{t-1}\beta<\frac{1.22p}{1-0.32}<2p,}
\]
\[
\begin{aligned}
d_L
&\leq
\sum_{t=L+1}^{\infty}0.4^{t-1}\beta
+
1.2\cdot0.9^{L-1}
+
0.3
\sum_{t=2}^{L}
4\cdot0.4^{t-1}0.9^{L-t}
\leq
4\cdot0.9^L .
\end{aligned}
\qedhere\]
\end{proof}

\subsection{Proof of \Cref{thm:fptas-counting-LLL} (Deterministic Approximate Counting)}
Enumerate the constraints in $\+C$ as $\+C=\{c_1,c_2,\ldots,c_m\}$ and define $\+C_i=\set{c_1,\ldots,c_i}$ for $0\le i\le m$. In particular, $\+C_0=\emptyset$.

We have the telescoping identity
\begin{align}\label{eq:reduction-Z-to-r}
Z_{\Phi}
=\tp{\prod\limits_{v\in V}|Q_v|}\Pr{\+C}
=\tp{\prod\limits_{v\in V}|Q_v|}\prod\limits_{i=1}^{m}\Pr{c_i\mid \+C_{i-1}}
=\tp{\prod\limits_{v\in V}|Q_v|}\cdot \prod\limits_{i=1}^{m}\tp{1-r_{\+C_i,c_i}}.
\end{align}
We then estimate $Z_{\Phi}$ by
\[
\widehat{Z}_{\Phi}=\tp{\prod\limits_{v\in V}|Q_v|}\cdot \prod\limits_{i=1}^{m}\tp{1-\hat{r}^{(L)}_{\+C_i,c_i}},
\]
where the truncation budget is set to
 $L= \lceil\log_{1.1}(40m\varepsilon^{-1})\rceil$.

By \Cref{lem:ratio-2-tree-correctness}, 
\[\abs{\hat{r}^{(L)}_{\+C_i,c_i}-r_{\+C_i,c_i}}\leq 4\cdot 0.9^L\leq 0.1\varepsilon m^{-1}.\]
Moreover, {$r_{\+C_i,c_i},\hat{r}^{(L)}_{\+C_i,c_i}\in [-2p,2p]$}.
Applying the mean value theorem to $f(x)=\log(1-x)$ gives
\[
\abs{\log{\frac{\hat{Z}_{\Phi}}{Z_{\Phi}}}}
\le \sum_{i=1}^m\abs{\log\frac{1-\hat{r}^{(L)}_{\+C_i,c_i}}{1-r_{\+C_i,c_i}}}
\leq \sum\limits_{i=1}^{m}\frac{\abs{\hat{r}^{(L)}_{\+C_i,c_i}-r_{\+C_i,c_i}}}{1-{2p}}\leq 2m\cdot 0.1\varepsilon m^{-1}=0.2\varepsilon,
\]
which implies $(1-\varepsilon)Z_{\Phi}\leq \hat{Z}_{\Phi}\leq (1+\varepsilon)Z_{\Phi}$.

Finally, $L$ can be chosen by exact rational comparisons, and each marginal
uses $q^{O(kDL)}$ oracle calls and $\widetilde O(q^{O(kDL)})$ time.
The final rational product has length at most the sum of its factor
lengths, so its arithmetic has polynomial overhead in $m$, $n\log q$,
and $q^{O(kDL)}$. Including preprocessing and using
$L=O(\log(m/\varepsilon))$ and $m\le n(D+1)$ gives the stated bound
$(nD/\varepsilon)^{O(kD\log q)}$ on time and oracle calls.

\section{Randomized Approximate Counting}\label{sec:randomized-counting}
\subsection{A 2-Tree Recursion for Randomized Estimation}\label{sec:unconditional-expansion}
We express the marginal probability as a sum over canonical $2$-trees, with coefficients given by unconditional joint probabilities.
Recall the definition of $\frin(T)$ in \Cref{lem:free-1-probability}.
For each tree $T$, the coefficient is the probability that every constraint in $T$ is violated and every constraint in $\frin(T)$ is satisfied.
These probabilities can be estimated using product assignments.

We first observe the following identity:
\begin{equation}\label{eq:cut-equation}
    \fr(T)\setminus \frin(T)=\+C\setminus \Gamma^{\leq 2}_{\+C}(T).
    \end{equation}
Here, $\Gamma^{\leq 2}_{\+C}(T)$ denotes the set of constraints at distance at most $2$ from $T$ in the dependency graph $G(\+C)$.
Indeed, by \Cref{Alg:gen-tree}, every constraint in the final rejected set $R(T)$ has distance at most $2$ from $T$. 
Combining this with the definitions of $\fr(T)$ and $\frin(T)$ gives the claimed identity.

The following lemma gives the $2$-tree recursion used by the estimator.
\begin{lemma}[$2$-tree recursion for randomized estimation]
\label{lemma:2-tree-expansion-recursion-alternate}

For $T\in \mathfrak{T}_{\+C,c_0}$, enumerate the constraints in
\begin{equation}\label{eq:2-tree-expansion-enumeration-alternate}
\+C\setminus\tp{\{c_0\}\cup \left(\fr(T)\setminus \frin(T)\right)}
=
\Gamma^{\leq 2}_{\+C}(T)\setminus \{c_0\}
=
\set{d^T_1,\ldots,d^T_{\iota^T}},
\end{equation}
where 
$\iota^{T}\defeq|\Gamma^{\leq 2}_{\+C}(T)|-1$.
For each $0\le i\le\iota^{T}$, define 
\[
\+D_i^{T} \defeq \tp{\+C\setminus \Gamma^{\leq 2}_{\+C}(T)}\cup \set{d^T_1,\dots,d^T_i}.
\]
Then the marginal probability $r_{\+C,c_0}$ satisfies
\begin{equation}\label{eq:2-tree-expansion-recurrence-alternate}
r_{\+C,c_0}=\sum\limits_{T\in \mathfrak{T}_{\+C,c_0} }(-1)^{\abs{T}-1} \Pr{ \tp{  \bigwedge\limits_{c\in  T }\neg c}\land \frin(T)}\cdot \prod\limits_{i=1}^{\iota^{T}}\frac{1}{1-r_{\+D^{T}_i,d^{T}_i}}.
\end{equation}
\end{lemma}

\begin{proof}
    We regroup the inclusion-exclusion expansion in \eqref{eq:inclusion-exclusion} as
\begin{equation}\label{eq:2-tree-expansion-refined-alternate}
\begin{aligned}
    r_{\+C,c_0} =& \;\; \frac{1}{\Pr{\+C\setminus\{c_0\}}}\sum\limits_{T\in \mathfrak{T}_{\+C,c_0}}(-1)^{\abs{T}-1}\sum\limits_{S\subseteq \fr(T)}(-1)^{\abs{S}}\Pr{\bigwedge\limits_{c\in S \cup T}\neg c}\\
    =& \;\; \frac{1}{\Pr{\+C\setminus\{c_0\}}}\sum\limits_{T\in \mathfrak{T}_{\+C,c_0}}(-1)^{\abs{T}-1}\Pr{  \tp{  \bigwedge\limits_{c\in  T }\neg c}\land \fr(T)}\\
    (\text{by \eqref{eq:cut-equation}})\quad=& \;\; \sum\limits_{T\in \mathfrak{T}_{\+C,c_0}}(-1)^{\abs{T}-1}\Pr{ \tp{\bigwedge\limits_{c\in  T }\neg c}\land  \frin(T) \,\middle|\, \+C\setminus \Gamma^{\leq 2}_{\+C}(T)}\cdot \frac{\Pr{\+C\setminus \Gamma^{\leq 2}_{\+C}(T)}}{\Pr{\+C\setminus\{c_0\}}}\\
    =& \;\; \sum\limits_{T\in \mathfrak{T}_{\+C,c_0}}(-1)^{\abs{T}-1}\Pr{ \tp{\bigwedge\limits_{c\in  T }\neg c}\land  \frin(T)}\cdot \frac{\Pr{\+C\setminus\Gamma^{\leq 2}_{\+C}(T)}}{\Pr{\+C\setminus\{c_0\}}}.
    \end{aligned}
\end{equation}
The last equality holds because every constraint in
$\+C\setminus\Gamma_{\+C}^{\le2}(T)$ has scope disjoint from
every constraint in $T\cup\frin(T)$. Hence its satisfaction event
is independent of
$\bigwedge_{c\in T}\neg c\land\frin(T)$ under $\mathbb P$.

By the definition of $\+D_i^T$, the remaining ratio telescopes:
\[
\frac{\Pr{\+C\setminus\Gamma^{\leq 2}_{\+C}(T)}}{\Pr{\+C\setminus\{c_0\}}} = \frac{Z(\+D_0^T)}{Z(\+C \setminus\{c_0\})} = \prod_{i = 1}^{\iota^T} \frac{Z(\+D_{i-1}^T)}{Z(\+D_i^T)} = \prod_{i = 1}^{\iota^T} \frac{1}{1 - r_{\+D^{T}_i,d^{T}_i}}.
\]
Substituting the above identity into \eqref{eq:2-tree-expansion-refined-alternate} gives the claimed $2$-tree recursion in \eqref{eq:2-tree-expansion-recurrence-alternate}.
\end{proof}

\subsection{Randomized Marginal Estimation}\label{sec:randomized-marginal-estimation}
For each constraint $c$, write $p_c=\Pr{\neg c}$.
We estimate the recursion in \Cref{lemma:2-tree-expansion-recursion-alternate} by sampling a canonical $2$-tree $T$.
Each constraint tested by \Cref{Alg:gen-tree} is accepted if its independent local assignment violates it.
The probability of selecting $T$ includes a factor $1-p_b$ for each rejected constraint $b\in R(T)$.
We compensate for these factors using independent unbiased estimates of $1/(1-p_b)$.
We also recursively estimate the reciprocal terms associated with the instances $\+D_i^T$.

Fix the order inherited from the input. For $T\in\mathfrak T_{\+C,c_0}$, let
\[
R(T)=\+C\setminus\bigl(T\cup\fr(T)\bigr)
\]
be the rejected set of the canonical extractor. All distances and free sets below are computed in the current induced graph $G(\+C)$.

\begin{algorithm}[htbp]
\caption{$\rrat(\+C,c_0)$: Randomized Estimator of the Marginal Probability $r_{\+C,c_0}$.}
\label{Alg:random-marginal-estimate}
\SetKwInOut{Input}{Input}
\SetKwInOut{Output}{Output}
\Input{A constraint set $\+C$ and a constraint $c_0\in\+C$.}
\Output{A random estimator $\widehat R_{\+C,c_0}$ of $r_{\+C,c_0}$.}
Associate independent local configurations $\sigma_c\sim\mathbb P_{\vbl(c)}$ with $c\in\+C$\label{Line:random-marginal-estimate-init}\;
\lIf{$\sigma_{c_0}$ satisfies $c_0$}{\Return $0$}
$T\gets\tr_{\+C,c_0}(\{c\in\+C:\sigma_c\text{ violates }c\})$\label{Line:random-marginal-estimate-choose}\;
$R\gets R(T)$\;
Extend $(\sigma_c)_{c\in T}$ by fresh product values to a configuration $X$ on $\vbl(T\cup\frin(T))$\label{Line:random-marginal-estimate-completion}\;
\lIf{$X$ violates a constraint in $\frin(T)$}{\Return $0$}
$Y\gets(-1)^{|T|-1}$\;
\ForEach{$b\in R$}{
    $Y\gets Y\cdot\reci(\{b\},b)$\;
}
Let $\iota^T,d_i^T,\+D_i^T$ be as defined in \Cref{lemma:2-tree-expansion-recursion-alternate}\;
\For{$i=1,\ldots,\iota^T$}{
    $Y\gets Y\cdot\reci(\+D_i^T,d_i^T)$\;
}
\Return $Y$\;
\end{algorithm}

The configurations in Line~\ref{Line:random-marginal-estimate-init} are independent even when their variable sets overlap.
Only those indexed by $T$ are retained for the completion; their scopes are disjoint. Every other coordinate of $X$ is drawn once from its product marginal. All reciprocal calls use fresh randomness, independently of the exploration and of one another.

\begin{remark}[Efficient Implementation]
\label{remark:efficient-implementation}
\emph{Preprocessing for representing recursive subinstances.}
Build the input dependency lists once in $\widetilde O(n+kmD)$ time
and share them across recursive calls. A Boolean indicator per constraint
records membership in the current family.
Process the child families $\+D_i^T$ in order, changing only the
necessary indicators and restoring changes on return.
Each radius-two neighborhood in $G(\+C)$ can be built in $O(D^2)$ graph operations using a two-step search restricted to the current family.

\emph{Generating and testing the sampled tree.}
Draw each independent local configuration $\sigma_c$ only when needed.
Test $c_0$ first and return zero if satisfied; otherwise run
\Cref{Alg:gen-tree}, accepting each examined $c$ when $\sigma_c$
violates it and recording rejections in $R(T)$.
By the principle of deferred decisions, this samples $T$ with the law specified in
\Cref{Alg:random-marginal-estimate}.
Obtain $\frin(T)$ from neighbors of $T$ by excluding $T\cup R(T)$.
Retain the configurations indexed by $T$, complete $X$ on
$\vbl(T\cup\frin(T))$ with fresh uniform values, and test
$\frin(T)$, without scanning the entire family.
\end{remark}

\begin{remark}[Construction of the Marginal Estimator]\label{remark:random-marginal-construction}
The construction of the marginal estimator in \Cref{Alg:random-marginal-estimate} consists of the following two steps:
\begin{itemize}
    \item It randomly selects a $2$-tree $T\in\mathfrak T_{\+C,c_0}$ and checks whether $X$ satisfies every constraint in $\frin(T)$. For each $T$, these two events occur jointly with probability
    \[
    \Pr{\left(\bigwedge_{c\in T}\neg c\right)\land\frin(T)}
    \prod_{b\in R(T)}(1-p_b).
    \]
    \item Whenever the events in the first step have positive joint probability, conditioning on them gives an unbiased estimate of
    \[
    (-1)^{|T|-1}
    \left(\prod_{b\in R(T)}\frac{1}{1-p_b}\right)
    \left(\prod_{i=1}^{\iota^T}\frac{1}{1-r_{\+D_i^T,d_i^T}}\right).
    \]
    Each reciprocal factor is estimated independently using $\reci$.
\end{itemize}
By the $2$-tree recursion~\eqref{eq:2-tree-expansion-recurrence-alternate},
the estimator is unbiased: multiplying the conditional expectation
in the second step by the joint probability in the first step
cancels the factors $1-p_b$, and summing over all $T$ gives the expectation $\oE[\widehat R_{\+C,c_0}]=r_{\+C,c_0}$.
\end{remark}

\subsubsection{The Reciprocal Estimator}
For independent copies $Y_1,Y_2,\ldots$ of a marginal estimate with mean $r$, the product series
\[
1+Y_1+Y_1Y_2+Y_1Y_2Y_3+\cdots
\]
has mean $(1-r)^{-1}$, provided the series is integrable. Here each marginal estimate is zero unless its initial product assignment violates the root constraint. The series therefore terminates at its first zero factor: every subsequent product is then exactly zero.

\begin{algorithm}[htbp]
\caption{$\reci(\+C,c_0)$: Unbiased Reciprocal Estimator.}
\label{Alg:reciprocal-estimate}
\SetKwInOut{Input}{Input}
\SetKwInOut{Output}{Output}
\Input{A constraint set $\+C$ and a constraint $c_0\in\+C$.}
\Output{A random estimator $W$ of $(1-r_{\+C,c_0})^{-1}$.}
$W\gets1$; $P\gets1$\;
\While{$P\ne0$}{
    $P\gets P\cdot\rrat(\+C,c_0)$, using fresh randomness\nllabel{ln:count-recip-iteration}\;
    $W\gets W+P$\;
}
\Return $W$\;
\end{algorithm}

For a singleton family, $\rrat(\{b\},b)$ is the indicator that a product assignment violates $b$. Hence $\reci(\{b\},b)$ is simply the number $N_b$ of independent product trials through the first satisfying assignment, including that last trial. In particular,
\begin{equation}\label{eq:singleton-reciprocal-moments}
\E{N_b}=\frac1{1-p_b},\qquad
\E{N_b^2}=\frac{1+p_b}{(1-p_b)^2}.
\end{equation}
The singleton calls in \Cref{Alg:random-marginal-estimate} implement these waiting counts. Their randomness is independent of the satisfying configurations observed at rejected exploration decisions.

Unbiased reciprocal estimation and, more generally, estimation of power series in an expectation are standard Monte Carlo techniques; see, for example, \cite{blanchet2015unbiased}.

\begin{lemma}[Zero-terminated reciprocal series]\label{lem:reciprocal-correctness}
Let $Y$ be a random variable satisfying
\[
\E{Y}=r\in[0,1),\qquad \E{Y^2}=M<1,\qquad
\tau=\pr{Y\ne0}<1.
\]
Let $Y_1,Y_2,\ldots$ be independent copies of $Y$.
Set $P_0=1$, $P_j=\prod_{h=1}^jY_h$, and $W=\sum_{j\ge0}P_j$, stopping at the first index $j$ with $Y_j=0$. Then
\begin{equation}\label{eq:inverse-moments}
\E{W}=\frac1{1-r},\qquad
\E{W^2}=\frac{1+r}{(1-r)(1-M)}.
\end{equation}
The expected number of samples of $Y$ is $(1-\tau)^{-1}$. If generating one sample has expected cost $C$, the expected total sampling cost is $C/(1-\tau)$.
\end{lemma}
\begin{proof}
The first zero occurs almost surely, and $\opr[P_j\ne0]=\tau^j$. This gives the expected number of samples. Since $\E{|Y|}\le\sqrt M<1$, the series of absolute first moments converges. Moreover, for $i\le j$,
\[
\E{P_iP_j}=M^ir^{j-i},\qquad
\E{|P_iP_j|}=M^i\bigl(\E{|Y|}\bigr)^{j-i}.
\]
The latter expression is summable over all pairs $i,j$, justifying both expectation interchanges. Thus
\[
\E{W}=\sum_{j\ge0}r^j=\frac1{1-r},\qquad
\E{W^2}=\sum_{i\ge0}M^i\left(1+2\sum_{h\ge1}r^h\right)
=\frac{1+r}{(1-r)(1-M)}.
\]
For the cost calculation, let $K_j$ be the cost of generating the
$j$-th independent copy of $Y$ using fresh randomness, so $\E{K_j}=C$.
This copy is requested precisely when $P_{j-1}\ne0$.
That event depends only on the preceding copies and is independent
of $K_j$, even if $K_j$ depends on $Y_j$.
Thus, by nonnegativity of the costs,
\[
\E{\sum_{j\ge1}\one{\{P_{j-1}\ne0\}}K_j}
=\sum_{j\ge1}\tau^{j-1}\E{K_j}
=\frac{C}{1-\tau}.\qedhere
\]
\end{proof}

\subsubsection{Correctness and Second Moment}
For a rooted $2$-tree $T$, the probability of selecting it and passing the completion test is
\begin{equation}\label{eq:single-tree-selection}
\pr{T \text{ is selected and its completion passes}}
=\Pr{\left(\bigwedge_{c\in T}\neg c\right)\land\frin(T)}\prod_{b\in R(T)}(1-p_b).
\end{equation}
Indeed, the canonical transcript accepts exactly $T$ and rejects exactly $R(T)$. Its independent local tests have probability $\prod_{c\in T}p_c\prod_{b\in R(T)}(1-p_b)$. Conditional on them, the retained configurations are independent product configurations conditioned on their own violations. Extending them and testing $\frin(T)$ contributes $\Pr{\left(\bigwedge_{c\in T}\neg c\right)\land\frin(T)}/\prod_{c\in T}p_c$. A tree containing a constraint with $p_c=0$ has zero probability and can be omitted from this calculation.

We use the following counting bound and include a proof for completeness.
\begin{lemma}[Rooted connected sets{~\cite[Lemma~2.1(c)]{borgs2013left}}]
\label{lem:rooted-tree-coefficients}
In a graph of maximum degree $\Delta\ge1$, the number of connected vertex sets of size $t$ containing a specified vertex is at most $\upsilon_t\Delta^{t-1}$, where $\upsilon_t=\frac{t^{t-1}}{t!}$.
\end{lemma}
\begin{proof}
For $t\ge2$, choose one spanning tree of each set and label its root $1$. Each set has $(t-1)!$ labelings. Cayley's formula gives $t^{t-2}$ trees on these labels, with at most $\Delta^{t-1}$ choices of their vertex images after the root is fixed. Dividing by $(t-1)!$ proves the count; $t=1$ is immediate.
\end{proof}

\begin{lemma}[Unbiasedness and second moment]\label{lem:random-marginal-estimate-correctness}
Let $\Phi=(V,\+Q,\+C)$ satisfy \Cref{cond:sym-counting-LLL}. For every $c_0\in\+C$, \Cref{Alg:random-marginal-estimate} terminates almost surely and returns an estimator $\widehat R_{\+C,c_0}$ satisfying
\begin{equation}\label{eq:single-tree-marginal-moments}
\E{\widehat R_{\+C,c_0}}=r_{\+C,c_0},\qquad
\E{\widehat R_{\+C,c_0}^2}\le3p_{c_0},\qquad
\pr{\widehat R_{\+C,c_0}\ne0}\le p_{c_0}.
\end{equation}
The assertions also hold for every subformula $(V,\+Q,\+D)$ with $\+D\subseteq\+C$.
\end{lemma}
\begin{proof}
We induct on $|\+C|$. On a singleton, the estimator is the product-violation indicator. In general, all reciprocal calls in \Cref{Alg:random-marginal-estimate} concern strict subfamilies: this is immediate for the singleton calls, and each $\+D_i^T$ omits $c_0$. The means of the recursive marginal estimates are at most $1.2 p$, their second moments are at most $3p<1$, and their probabilities of being nonzero are at most $p<1$. The induction hypothesis and \Cref{lem:reciprocal-correctness} therefore give almost-sure termination and finite first and second moments for every factor. A nonzero output requires the initial assignment to violate $c_0$, which occurs with probability $p_{c_0}$. Thus, $\opr[\widehat R_{\+C,c_0}\ne0]\le p_{c_0}$.

By \eqref{eq:single-tree-selection}, \eqref{eq:singleton-reciprocal-moments}, and independence, the expected contribution of $T$ is
\[
(-1)^{|T|-1}\Pr{\left(\bigwedge_{c\in T}\neg c\right)\land\frin(T)}
\left(\prod_{b\in R(T)}
    (1-p_b)\cdot\frac{1}{1-p_b}\right)
\left(\prod_{i=1}^{\iota^T}
    \frac{1}{1-r_{\+D_i^T,d_i^T}}\right).
\]
Summing over all $T$ gives $\oE[\widehat R_{\+C,c_0}]=r_{\+C,c_0}$ by \Cref{lemma:2-tree-expansion-recursion-alternate}.

Finally, we bound the second moment. Set
\[
\Gamma=(D+1)^2,\qquad
B=\frac{1+1.2p}{(1-1.2p)(1-3p)},\qquad
H=B^\Gamma \cdot \left(\frac{1+p}{1-p}\right)^{D^2}.
\]
By \Cref{lem:marginal-bound}, the mean of every recursive marginal estimate is at most $1.2p$. Its second moment is at most $3p$ by induction, so \Cref{lem:reciprocal-correctness} bounds the second moment of each reciprocal estimate by $B$.
The events corresponding to distinct selected trees are disjoint. Let $W_i^T$ denote the output of $\reci(\+D_i^T,d_i^T)$. Then
\begin{equation}\label{eq:single-tree-second-moment}
\E{\widehat R_{\+C,c_0}^2}
=\sum_T \Pr{\left(\bigwedge_{c\in T}\neg c\right)\land\frin(T)}\prod_{b\in R(T)}\frac{1+p_b}{1-p_b}
                 \prod_{i=1}^{\iota^T}\E{(W_i^T)^2}.
\end{equation}
For $t=|T|$, disjointness of the accepted scopes gives $\Pr{\left(\bigwedge_{c\in T}\neg c\right)\land\frin(T)}\le p_{c_0}p^{t-1}$. Also $|R(T)|\le D^2t$ and $\iota^T\le\Gamma t$. Applying \Cref{lem:rooted-tree-coefficients} to $G^2(\+C)$, whose max-degree is at most $D^2$, yields
\begin{equation}\label{eq:single-tree-tree-sum}
\E{\widehat R_{\+C,c_0}^2}
\le p_{c_0}H\sum_{t\ge1}\upsilon_t(pD^2H)^{t-1}.
\end{equation}
When $p=0$ the output is identically zero, so assume $p>0$. We bound the series in \eqref{eq:single-tree-tree-sum} using Lagrange inversion. For this purpose, we compare $pD^2H$ with $3pD^2\exp(-3pD^2)$. Under \Cref{cond:sym-counting-LLL}, $p\Gamma<1/10$ and $p<1/40$. Using $\log(1+x)\le x$ and $-\log(1-x)\le x/(1-x)$,
\begin{align*}
pD^2H
&\le pD^2\exp\left[p\Gamma\left(\frac2{1-p}+3+\frac{2.4}{1-1.2p}+\frac3{1-3p}\right)-3pD^2\right]\\
&<pD^2\exp\left(10.8p\Gamma-3pD^2\right)
<3pD^2\exp(-3pD^2).
\end{align*}
Lagrange inversion gives the identity $\sum_{t\ge1}\upsilon_t\bigl(x\exp(-x)\bigr)^t=x$ for any $0 \leq x < 1$.
Applying this identity with $x=3pD^2<1$ and using the bound above, we obtain
\[
\sum_{t\ge1}\upsilon_t(pD^2H)^t
\le\sum_{t\ge1}\upsilon_t\bigl(3pD^2\exp(-3pD^2)\bigr)^t
=3pD^2.
\]
Substituting this bound into \eqref{eq:single-tree-tree-sum} proves the second-moment bound and completes the induction.
\end{proof}

\subsubsection{Efficiency of the Marginal Estimator}
We next analyze the efficiency of \Cref{Alg:random-marginal-estimate} under the implementation described in \Cref{remark:efficient-implementation}. The key observation is that the expected number of recursive calls generated by each invocation is bounded by a constant strictly smaller than $1$.
\begin{lemma}[Efficiency of the randomized marginal estimator]\label{lem:random-marginal-estimate-efficiency}
Let $\Phi=(V,\+Q,\+C)$ satisfy \Cref{cond:sym-counting-LLL}.
For any $c_0\in\+C$, after the shared preprocessing in
\Cref{remark:efficient-implementation}, the procedure $\rrat(\+C,c_0)$
uses $\poly(k,q,D)$ expected evaluation-oracle calls and
$\widetilde O(\poly(k,q,D))$ expected time.
\end{lemma}
\begin{proof}
We first bound the expected number of recursive calls. Let $T=\emptyset$ if the initial assignment satisfies $c_0$; otherwise, let $T$ be the selected $2$-tree. The root belongs to $T$ with probability $p_{c_0}$. Every other constraint in $T$ is accepted after a test, and at most $D^2|T|$ constraints are tested after the root. Conditional on the preceding tests, each is accepted with probability at most $p$. Thus $\E{|T|}\le p_{c_0}+pD^2\E{|T|}$, and hence
\begin{equation}\label{eq:single-tree-size}
\E{|T|}\le\frac{p_{c_0}}{1-pD^2}.
\end{equation}
Each selected $T$ gives at most $\Gamma|T|$ calls to $\reci(\+D_i^T,d_i^T)$, where $\Gamma=(D+1)^2$. By \Cref{lem:reciprocal-correctness}, each such call makes at most $(1-p)^{-1}$ calls to the marginal estimator in expectation. Therefore, the expected number of recursive calls generated is at most
\[
\beta:=\frac{p\Gamma}{(1-p)(1-pD^2)}
<\frac{1/10}{(39/40)(9/10)}=\frac{40}{351}<\frac18.
\]
The calls $\reci(\{b\},b)$ for $b\in R(T)$ are implemented by independent trials until $b$ is satisfied, and are included in the cost of the current invocation.

We next count graph operations and oracle calls in one invocation, excluding its recursive calls. Under \Cref{remark:efficient-implementation}, constructing $T$, $R(T)$, $\frin(T)$, and the child instances uses $\poly(k,q,D,|T|)$ graph operations. Each $T$ is selected with probability at most $p^{|T|}$, and the number of $2$-trees of size $t$ containing $c_0$ is at most $(\e D^2)^{t-1}$. Moreover, $|R(T)|\le D^2|T|$, and each call $\reci(\{b\},b)$ uses at most $(1-p)^{-1}$ trials in expectation. Including the initial test at $c_0$ and using $\e pD^2\le1/4$, the expected cost is at most
\[
\poly(k,q,D)+\sum_{t\ge1}\poly(k,q,D,t)(\e D^2)^{t-1}p^t
\le\poly(k,q,D)\left(1+\sum_{t\ge1}\poly(t)4^{-t}\right)
=\poly(k,q,D).
\]

Since each invocation generates at most $\beta<1$ recursive calls
in expectation, the expected total number of invocations is at most
$(1-\beta)^{-1}$.  The preceding local work bound holds uniformly
over recursive subinstances and gives the claimed expected oracle
bound.
\Cref{remark:counting-bit-certificate} accounts for integer arithmetic, data access, and exact random choices, completing the stated expected running-time bound.
\end{proof}

\subsection{Proof of \Cref{thm:fpras-counting-LLL} (Randomized Approximate Counting)}\label{sec:random-count-proof}
We first construct an estimator $\widetilde Z_\Phi$ of $Z_\Phi$ based on the randomized marginal estimator in \Cref{Alg:random-marginal-estimate}. The estimator has the desired accuracy guarantee and bounded expected total cost. We then apply a standard truncation argument to obtain a worst-case cost bound.

We recall the telescoping identity in \eqref{eq:reduction-Z-to-r}, which relates $Z_\Phi$ to the marginal probabilities $r_{\+C_i,c_i}$:
\[
Z_{\Phi}=\left(\prod_{v\in V}|Q_v|\right)
              \prod_{i=1}^{m}(1-r_{\+C_i,c_i}),
\]
where $\+C_i=\{c_1,\ldots,c_i\}$ under an arbitrary ordering $\+C=\{c_1,\ldots,c_m\}$.
Set $L=\lceil 8m(D+1)^{-2}\varepsilon^{-2}\rceil$.
For each $1\le i\le m$, invoke $\rrat(\+C_i,c_i)$ independently $L$ times, and denote the outputs by $\widehat R_{\+C_i,c_i}^{(1)},\ldots,\widehat R_{\+C_i,c_i}^{(L)}$. Define their average by $\overline R_{\+C_i,c_i}=L^{-1}\sum_{j=1}^L\widehat R_{\+C_i,c_i}^{(j)}$ and the estimator $\widetilde Z_\Phi$ is
\begin{equation}\label{eq:random-product-estimator}
\widetilde Z_{\Phi}=\left(\prod_{v\in V}|Q_v|\right) \prod_{i=1}^{m}(1-\overline R_{\+C_i,c_i}).
\end{equation}
We first show that $\widetilde Z_\Phi$ approximates $Z_\Phi$ within relative error $\varepsilon$ with probability at least $15/16$:
\begin{equation}\label{eq:estimator-Z-approx}
\pr{(1-\varepsilon)Z_\Phi\le\widetilde Z_\Phi\le(1+\varepsilon)Z_\Phi}
\ge\frac{15}{16}.
\end{equation}
All $mL$ marginal invocations use independent randomness. By \Cref{lem:random-marginal-estimate-correctness}, each average has mean $r_{\+C_i,c_i}$ and variance at most $3p_{c_i}/L\le3p/L$. Independence therefore gives $\oE[\widetilde Z_\Phi]=Z_\Phi$. By \Cref{lem:marginal-bound}, $r_{\+C_i,c_i}\le1.2p<0.03$, so $3/(1-r_{\+C_i,c_i})^2<4$. Hence,
\[
\frac{\E{\widetilde Z_\Phi^2}}{Z_\Phi^2} =\prod_{i=1}^{m}\left(1+ \frac{\Var{\overline R_{\+C_i,c_i}}}{(1-r_{\+C_i,c_i})^2}\right)\le\exp\left(\frac{4pm}{L}\right)
\le\exp\left(\frac{\varepsilon^2}{20}\right).
\]
Here we used $p(D+1)^2<1/10$ and the choice of $L$. Since $\exp(\varepsilon^2/20)-1<\varepsilon^2/16$, the relative variance is less than $\varepsilon^2/16$. Applying Chebyshev's inequality gives \eqref{eq:estimator-Z-approx}.

By \Cref{lem:random-marginal-estimate-efficiency}, the
$mL=O(nD+n^2\varepsilon^{-2})$ marginal invocations use
$\poly(k,q,D)mL$ expected oracle calls.
\Cref{remark:counting-bit-certificate} also accounts for the final product and preprocessing, giving total expected running time $\widetilde O(\poly(k,D,q)(n/\varepsilon)^2)$ for the untruncated estimator.

Finally, we apply a standard truncation argument. For oracle calls
and running time, respectively, fix upper bounds on their expected
costs of the orders established above. Stop and return $0$ if either
cost exceeds $16$ times its respective bound; otherwise return
$\widetilde Z_\Phi$.
By Markov's inequality, each limit is exceeded with probability at
most $1/16$. Together with \eqref{eq:estimator-Z-approx}, a union bound
gives success probability at least
$\frac{15}{16}-\frac{1}{16}-\frac{1}{16}
=\frac{13}{16}>\frac34$.
The two cutoffs give the claimed worst-case bounds on oracle calls
and running time. This proves
\Cref{thm:fpras-counting-LLL}.

\begin{remark}[Bit complexity of randomized counting]\label{remark:counting-bit-certificate}
By \Cref{lem:random-marginal-estimate-correctness} and
\Cref{lem:reciprocal-correctness}, the integer outputs have bounded second moments, hence bounded fixed-order bit-length moments. Multiplying the reciprocal factors for $T$ costs $\poly(D,|T|)$ expected bit operations. For the analysis, pre-sample a reciprocal call's inputs: iteration $j$ is reached with probability at most $p^{j-1}$, and its running lengths have moments polynomial in $j$. Cauchy--Schwarz bounds its expected arithmetic cost, excluding input generation, by $\poly(j)p^{(j-1)/2}$. Summing and applying the preceding tree-size and recursive-call bounds gives the claimed marginal time. Charge fresh calls before observing their outputs; indexed storage and exact uniform draws add logarithmic factors.

For \eqref{eq:random-product-estimator}, partial batch sums and factors have $r$th bit-length moments $O_r(\log^r(L+1))$, since the expected sum of absolute inputs is $O(L)$. Ordinary integer arithmetic, without fraction reduction, therefore computes the final estimator in $\widetilde O(mL+m^2+n^2)$ expected time, excluding marginal generation.
\end{remark}

\section{Near-Linear Time Exact Sampling}\label{sec:samp-algorithm}

We extend the $2$-tree recursion framework to exact sampling.
Rather than estimating only a marginal probability, our goal is to
construct a randomized transport from
$\mu_{\+C\setminus\{c_0\}}$ to $\mu_{\+C}$.

Recall $\frin(T)$ from \Cref{lem:free-1-probability} and the $2$-tree recursion \eqref{eq:2-tree-expansion-recurrence-alternate} in
\Cref{lemma:2-tree-expansion-recursion-alternate}.
For $T\in\mathfrak T_{\+C,c_0}$, let
\[
 \+E_T\defeq\tp{\bigwedge_{c\in T}\neg c}\land\frin(T).
\]
For a measure $\mu$ and an event $E$, write $\mu\one{E}$ for the
unnormalized restriction of $\mu$ to $E$, defined by
\[
(\mu\one{E})(A)\defeq\mu(A\cap E), \quad\text{for every event }A.
\]
We apply the $2$-tree expansion to the
\emph{$c_0$-violation measure}
$\mu_{\+C\setminus\{c_0\}}\one{\neg c_0}$.

\begin{lemma}[$2$-tree expansion of the $c_0$-violation measure]\label{lem:distributional-2-tree-expansion}
For $c_0\in\+C$ and $\Sol_{\+C}\ne\varnothing$,
\begin{equation}\label{eq:2-tree-expansion-recurrence-distribution}
 \mu_{\+C\setminus\{c_0\}}\one{\neg c_0}
 =\sum_{T\in\mathfrak T_{\+C,c_0}}(-1)^{|T|-1}
 \Pr{\+E_T}
 \prod_{i=1}^{\iota^T}\frac1{1-r_{\+D_i^T,d_i^T}}
 \mu_{\+D_0^T}(\,\cdot\mid\+E_T).
\end{equation}
Here $d_i^T$, $\iota^T$ and $\+D_i^T$ are as defined in \Cref{lemma:2-tree-expansion-recursion-alternate}. 
Terms with $\Pr{\+E_T}=0$ are omitted.
\end{lemma}
\begin{proof}
For every $Y\in\+Q$, inclusion--exclusion in
\eqref{eq:inclusion-exclusion} gives
\begin{equation}\label{eq:2-tree-expansion-recurrence-distribution-pointwise}
\begin{aligned}
 \one{\neg c_0\land(\+C\setminus\{c_0\})}(Y)
 &=
 \sum_{c_0\in\+D\subseteq\+C}(-1)^{|\+D|-1}
 \one{\tp{\bigwedge_{c\in\+D}\neg c}}(Y)\\
 &=
 \sum_{T\in\mathfrak T_{\+C,c_0}}(-1)^{|T|-1}
 \one{\tp{\bigwedge_{c\in T}\neg c}\land\fr(T)}(Y)\\
 &=
 \sum_{T\in\mathfrak T_{\+C,c_0}}(-1)^{|T|-1}
 \one{\+E_T}(Y)\one{\+D_0^T}(Y).
\end{aligned}
\end{equation}
The second equality groups terms by
$T=\tr_{\+C,c_0}(\+D)$, and the last uses
\eqref{eq:cut-equation}.

For every $T$ with $\Pr{\+E_T}>0$, the constraints in
$\+D_0^T$ contain variables disjoint from $\vbl(T\cup\frin(T))$.
Thus
$\mu_{\+D_0^T}(\+E_T)=\Pr{\+E_T}$, and hence
$\one{\+E_T}(Y)\mu_{\+D_0^T}(Y)
=
\Pr{\+E_T}\,
\mu_{\+D_0^T}(\,\cdot\mid\+E_T)(Y)$. By telescoping this gives
\[
\begin{aligned}
\frac{\one{\+E_T}(Y)\one{\+D_0^T}(Y)}
     {Z(\+C\setminus\{c_0\})}
&=
\one{\+E_T}(Y)\mu_{\+D_0^T}(Y)
\prod_{i=1}^{\iota^T}
\frac{1}{1-r_{\+D_i^T,d_i^T}}\\
&=
\Pr{\+E_T}
\prod_{i=1}^{\iota^T}
\frac{1}{1-r_{\+D_i^T,d_i^T}}\,
\mu_{\+D_0^T}(\,\cdot\mid\+E_T)(Y).
\end{aligned}
\]
Dividing \eqref{eq:2-tree-expansion-recurrence-distribution-pointwise}
by $Z(\+C\setminus\{c_0\})$ and substituting this identity for each
summand proves
\eqref{eq:2-tree-expansion-recurrence-distribution}.
\end{proof}

\subsection{Exact Sampling via Dynamic Sampler}
A crucial subroutine is a \emph{dynamic sampler}: 
it maintains an exact sample from the target distribution as the CSP changes through constraint insertions and removals. 
Given $X\sim\mu_{\+C\setminus\{c_0\}}$ (or $X\sim\mu_{\+C}$), the incremental (or decremental) update produces $Y\sim\mu_{\+C}$ (or $Y\sim\mu_{\+C\setminus\{c_0\}}$).
This constraint-wise self-reducibility for sampling was used in~\cite{wang2024sampling}.

We realize these dynamic updates through the $2$-tree expansion of the
$c_0$-violation measure in
\eqref{eq:2-tree-expansion-recurrence-distribution}.
Its alternating signs lead to a transport from the $c_0$-violation
term and the even-tree terms to the odd-tree terms.
Call $T\in\mathfrak T_{\+C,c_0}$ \emph{even} if $|T|$ is even, and
\emph{odd} otherwise. For each $T\in\mathfrak T_{\+C,c_0}$, define
\[
 w_T\defeq
 \Pr{\+E_T}\prod_{i=1}^{\iota^T}
 \frac{1}{1-r_{\+D_i^T,d_i^T}},
 \quad\text{ and }\quad
 \nu_T\defeq\mu_{\+D_0^T}(\,\cdot\mid\+E_T)
 \quad\text{if } w_T>0.
\]
When $w_T=0$, we interpret $w_T\nu_T$ as the zero measure.

Rearranging \eqref{eq:2-tree-expansion-recurrence-distribution} gives the following \emph{parity-balance} identity:
\begin{equation}\label{eq:samp-even-odd-balance}
 \mu_{\+C\setminus\{c_0\}}\one{\neg c_0}
 +\sum_{\text{even }T\in\mathfrak T_{\+C,c_0}}
   w_T\nu_T
 =
 \sum_{\text{odd }T\in\mathfrak T_{\+C,c_0}}
   w_T\nu_T.
\end{equation}
Thus, the $c_0$-violation measure is the odd-tree contribution minus the even-tree contribution.

We first explain the incremental update, which transforms $X\sim\mu_{\+C\setminus\{c_0\}}$ into $Y\sim\mu_{\+C}$. 
Following the rejection-sampling idea, we test whether $X$ satisfies the new constraint $c_0$:
\begin{itemize}
 \item If $X$ satisfies $c_0$, this branch has measure $(1-r_{\+C,c_0})\mu_{\+C}$ and $X$ may be returned immediately.
 \item Otherwise, it has the $c_0$-violation measure $\mu_{\+C\setminus\{c_0\}}\one{\neg c_0}$, which must be repaired into a scalar multiple of $\mu_{\+C\setminus\{c_0\}}$.
\end{itemize}

The parity-balance identity \eqref{eq:samp-even-odd-balance} suggests how to perform this repair.
Its algorithmic realization requires three ingredients:
\begin{enumerate}
 \item \textbf{Weighted $2$-tree proposals.}
The procedure $\prop(\+C,c_0)$ returns each canonical $2$-tree $T$
with probability $w_T$, and otherwise returns $\varnothing$.

 \item\label{item-tree-conditioned-transport} \textbf{Recursive access to tree-conditioned laws.}
 For a proposed tree $T$, self-reducibility gives the forward and
 reverse transports
 \begin{align*}
     &\mu_{\+C\setminus\{c_0\}}
  \xrightarrow{\text{recursive removals}}
  \mu_{\+D_0^T}
  \xrightarrow{\text{local conditioning}}
  \nu_T;\\
  &\nu_T
  \xrightarrow{\text{local product redraw}}
  \mu_{\+D_0^T}
  \xrightarrow{\text{recursive insertions}}
  \mu_{\+C\setminus\{c_0\}}.
 \end{align*}
 Thus, recursion lets the update pass from its original law to
$\nu_T$, and later return.

 \item \textbf{Routing between $2$-tree parities.}
We work with tree-assignment pairs $(T,X)$. A tree term $w_T\nu_T$ in
\eqref{eq:samp-even-odd-balance} is represented by
$X\sim\nu_T$. The only exception is the $c_0$-violation term on the
left-hand side: we represent it by setting $T=\{c_0\}$, the root
$2$-tree, without requiring $X\sim\nu_{\{c_0\}}$.

A fair coin gives the $c_0$-violation term and the even-tree terms a
common factor of $1/2$. Their joint measure is then one half of the
left-hand side of \eqref{eq:samp-even-odd-balance}. The required
transport is
\[
 \underbrace{\text{$c_0$-violation measure}
 +\text{even-tree terms}}_{\text{source}}
 \xrightarrow{\ \mathsf{Route}\ }
 \underbrace{\text{odd-tree terms}}_{\text{target}}.
\]
The procedure $\mathsf{Route}$ realizes this transport through the
canonical $2$-tree construction. It preserves $X$ and changes $T$
where necessary.

\end{enumerate}

For the incremental update, when rejection sampling fails, a fair coin
pairs the failed measure with the even-tree terms. The forward
transports realize these terms, $\mathsf{Route}$ converts the mixture
into the odd-tree terms, and the reverse transports return them to a
multiple of $\mu_{\+C\setminus\{c_0\}}$. Repeating this
rejection-and-repair construction yields a sample from $\mu_{\+C}$.
The decremental update uses the same identity in reverse.
Algorithm~\ref{alg:samp-update} implements both updates in one
procedure, where $\mathsf{op}=+1$ denotes insertion and
$\mathsf{op}=-1$ denotes removal.

\begin{algorithm}[ht]
\caption{$\mathsf{Update}(\+C,c_0,\mathsf{op};X)$: Dynamic Sampler.}
\label{alg:samp-update}
\SetKwInOut{Input}{Input}\SetKwInOut{Output}{Output}
\Input{A constraint set $\+C$, $c_0\in\+C$,
$\mathsf{op}\in\{+1,-1\}$ indicating, respectively, insertion or
removal of $c_0$, and an initial assignment
$X\in\Omega_{\+C\setminus\{c_0\}}$ for $\mathsf{op}=+1$, or
$X\in\Omega_{\+C}$ for $\mathsf{op}=-1$.}
\Output{An assignment in $\Sol_{\+C}$ for $\mathsf{op}=+1$
 and in $\Sol_{\+C\setminus\{c_0\}}$ for $\mathsf{op}=-1$.}
\While{true}{
 $T\gets\prop(\+C,c_0)$\label{Line:samp-update-propose}\;
 \lIf{$\mathsf{op}=-1$, $T\ne\varnothing$, and $|T|$ is even\label{Line:samp-update-return-2}}{\Return $X$}
\eIf{a fair coin lands heads}{  
\lIf{$\mathsf{op}=-1$}{\textbf{continue}}
\KwSty{if} $X$ satisfies $c_0$\label{Line:samp-update-return-1}
\KwSty{then} \Return $X$ \KwSty{else} $T\gets\{c_0\}$\;
 }{
  \lIf{$\mathsf{op}=+1$ and ($T=\varnothing$ or $|T|$ is odd)}{\textbf{continue}}
  \lIf{$T=\varnothing$\label{Line:samp-update-return-3}}{\Return $X$}
  \lFor{$i=\iota^T,\ldots,1$}{
   $X\gets\mathsf{Update}(\+D_i^T,d_i^T,-1;X)$
    \label{Line:samp-update-remove}
  }
  \Repeat{$X$ satisfies $\frin(T)\label{Line:samp-update-repeat}$}{
   \label{Line:samp-update-draw}
   Redraw $X_{\vbl(T\cup\frin(T))}$ from
$\mathbb P_{\vbl(T\cup\frin(T))}
(\,\cdot\mid\bigwedge_{c\in T}\neg c)$\label{Line:samp-update-local}\;
  }
 }
 $(T,X)\gets\mathsf{Route}(T,X)$\label{Line:samp-update-change}\;
 \lIf{$\mathsf{op}=-1$ and $|T|$ is odd\label{Line:samp-update-continue}}{\textbf{continue}}
 Redraw $X_{\vbl(T\cup\frin(T))}$ from
$\mathbb P_{\vbl(T\cup\frin(T))}$\label{Line:samp-update-redraw}\;
 \lFor{$i=1,\ldots,\iota^T$}{
  $X\gets\mathsf{Update}(\+D_i^T,d_i^T,+1;X)$
   \label{Line:samp-update-insert}
 }
}
\end{algorithm}

The conditional draw in Lines~\ref{Line:samp-update-draw}--\ref{Line:samp-update-repeat} is realized by independent
rejection sampling on the disjoint scopes of the constraints in $T$,
followed by product draws on the remaining variables.

The following lemma establishes the correctness and efficiency of \Cref{alg:samp-update}.

\begin{lemma}[Correctness and efficiency of the dynamic sampler]
\label{lem:samp-insert-erase-completed}
Let $\Phi=(V,\+Q,\+C)$ satisfy \Cref{cond:sym-counting-LLL}, and let
$c_0\in\+C$. Then the following hold:
\begin{align*}
    X\sim\mu_{\+C\setminus\{c_0\}} &\implies \mathsf{Update}(\+C,c_0,+1;X)\sim\mu_{\+C};\\
    X\sim\mu_{\+C} &\implies \mathsf{Update}(\+C,c_0,-1;X)\sim\mu_{\+C\setminus\{c_0\}}.
\end{align*}
Under the respective input laws, both calls terminate almost surely;
assuming the preprocessing in \Cref{remark:efficient-implementation},
each call has expected oracle cost $O(1)$ and expected running time
$\widetilde O(kD^2)$.
\end{lemma}

\begin{proof}[Proof of \Cref{thm:fpaus-sampling-LLL}
(Near-Linear Time Exact Sampling), assuming \Cref{lem:samp-insert-erase-completed}]
We construct an exact sampler by successively adding constraints.
Recall the families $\+C_i$ from \eqref{eq:reduction-Z-to-r}, with
$c_1,\ldots,c_m$ in input order. Draw $X$ uniformly from $\+Q$.
For $i=1,\ldots,m$, replace $X$ by
$\mathsf{Update}(\+C_i,c_i,+1;X)$, using fresh randomness in each call.
Return the final assignment.

Each $\+C_i$ satisfies
\Cref{cond:sym-counting-LLL}.
Initially $X\sim\mu_\varnothing$. Since
$\+C_i\setminus\{c_i\}=\+C_{i-1}$,
\Cref{lem:samp-insert-erase-completed} gives $X\sim\mu_{\+C_i}$
after the $i$th update. The returned assignment therefore has
distribution $\mu_{\+C}$.

Each update has expected oracle cost $O(1)$ and expected running time
$\widetilde O(kD^2)$. The membership indicators start
empty and add $c_i$ before the $i$th update, taking $O(m)$ operations
in total. Including the initial draw and the preprocessing in
\Cref{remark:efficient-implementation} gives expected oracle cost
$O(m)$ and expected running time $\widetilde O(n+kmD^2)$.
This proves \Cref{thm:fpaus-sampling-LLL}.
\end{proof}

\subsection{Analysis of the Dynamic Sampler}

We first state the guarantees of the two subroutines used by the dynamic
sampler: the exact tree proposal $\prop$ and the routing procedure
$\mathsf{Route}$. 
They both assume the preprocessing in \Cref{remark:efficient-implementation}.

\begin{lemma}[Correctness and efficiency of $\prop$]
\label{lem:samp-tree-law}
Let $\Phi=(V,\+Q,\+C)$ satisfy \Cref{cond:sym-counting-LLL}, and let
$c_0\in\+C$. The procedure $\prop(\+C,c_0)$ returns either
$T\in\mathfrak T_{\+C,c_0}$ or $\varnothing$, with
\begin{equation}\label{eq:samp-tree-law}
 \pr{\prop(\+C,c_0)=T}=w_T
 \qquad
 \text{for every }T\in\mathfrak T_{\+C,c_0}.
\end{equation}
It returns $\varnothing$ with the remaining probability.
Every call has expected oracle cost
$O(1)$ and expected running time $\widetilde O(kD^2)$.
All probabilities and expectations in this lemma are over the internal randomness
of $\prop$.
\end{lemma}

The construction and analysis of $\prop(\+C,c_0)$ are deferred to
\Cref{sec:samp-propose}. In particular, the proof establishes
$\sum_{T\in\mathfrak T_{\+C,c_0}}w_T\leq1$ under \Cref{cond:sym-counting-LLL}.

To describe the routing step, define two unnormalized measures on
tree-assignment pairs $(T,X)$, where
$T\in\mathfrak T_{\+C,c_0}$ and $X\in\+Q$, by
\begin{align*}
\nu_{\mathrm{even}}(T,X)
&\defeq
\one{T\text{ even}}\,(w_T\nu_T)(X)
+\one{T=\{c_0\}}
\bigl(\mu_{\+C\setminus\{c_0\}}\one{\neg c_0}\bigr)(X),\\
\nu_{\mathrm{odd}}(T,X)
&\defeq
\one{T\text{ odd}}\,(w_T\nu_T)(X).
\end{align*}
Thus, $\nu_{\mathrm{even}}$ is the left-hand side of
\eqref{eq:samp-even-odd-balance} expressed as a measure on
tree-assignment pairs: its $c_0$-violation term is represented by
pairs with $T=\{c_0\}$. Likewise, $\nu_{\mathrm{odd}}$ represents the
right-hand side. By \eqref{eq:samp-even-odd-balance}, the two measures
have the same total mass.

\begin{lemma}[Correctness and efficiency of $\mathsf{Route}$]
\label{lem:samp-parity-pairing}
Fix a constraint set $\+C$ and $c_0\in\+C$ with
$\Sol_{\+C\setminus\{c_0\}}\ne\varnothing$.
The procedure $\mathsf{Route}(T,X)$ in \Cref{alg:samp-route} returns
a pair $(T',X)$, leaves the assignment unchanged, and transports these
measures as follows:
\begin{enumerate}
\item with input measure $\nu_{\mathrm{even}}$, its output has
measure $\nu_{\mathrm{odd}}$;
\item with input measure $\nu_{\mathrm{odd}}$, its output has
measure $\nu_{\mathrm{even}}$.
\end{enumerate}
For every fixed pair in either corresponding input support, a call has
oracle cost $O(D^2|T|)$ and expected running time
$\widetilde O(kD^4|T|)$, where the expectation is over the internal
randomness of $\mathsf{Route}$.
\end{lemma}

The construction and analysis of $\mathsf{Route}$ are deferred to
\Cref{sec:samp-route}.

\begin{proof}[Proof of \Cref{lem:samp-insert-erase-completed},
assuming Lemmas~\ref{lem:samp-tree-law} and~\ref{lem:samp-parity-pairing}]
Write $\mu_-\defeq\mu_{\+C\setminus\{c_0\}}$,
$\mu_+\defeq\mu_{\+C}$, and $r\defeq r_{\+C,c_0}$. We induct on
$|\+C|$.

\emph{Validity and termination.}
Every recursive family omits $c_0$, hence is smaller than $\+C$. For
every processed $T$,
$\+D_{\iota^T}^T=\+C\setminus\{c_0\}$ and
$\+D_{i-1}^T=\+D_i^T\setminus\{d_i^T\}$, so the recursive calls have
the required input supports. After the removals, the conditional draw
puts $X$ in the support of $\nu_T$; \Cref{lem:samp-parity-pairing}
maps this to the opposite tree support. The product redraw and
recursive insertions restore $\Sol_{\+C\setminus\{c_0\}}$, while the
exceptional decremental root case already lies in this set. Thus every
continuation is valid. By induction, the recursive calls terminate
almost surely, and \Cref{lem:samp-tree-law} gives the same for
proposals. Every tree reaching the local-draw loop has
$\Pr{\+E_T}>0$, so that loop also terminates.

\emph{One-pass measures.}
Consider an outer iteration entered with $\mu_-$. For insertion, the
failed-head branch together with tail branches proposing even trees
sends $\tfrac12\nu_{\mathrm{even}}$ to $\mathsf{Route}$. For removal,
tail branches proposing odd trees send
$\tfrac12\nu_{\mathrm{odd}}$. By \Cref{lem:samp-parity-pairing},
$\mathsf{Route}$ exchanges these measures. The redraw and recursive
insertions send every nonroot tree term $w_T\nu_T$ back to
$w_T\mu_-$; only the root term in the decremental direction continues
unchanged. Taking total masses in
\eqref{eq:samp-even-odd-balance} gives
$\sum_{\text{odd }T}w_T-\sum_{\text{even }T}w_T=r$.
For removal, the unchanged and restored branches contribute
$\tfrac12\mu_-$, while the exceptional root contributes
$\tfrac12\mu_-\one{\neg c_0}$. Together with the direct-return
branches, this gives
\[
\begin{array}{c|c|c}
 & \text{returning measure} & \text{continuing measure}\\ \hline
 \mathsf{op}=+1
 & \frac{1-r}{2}\mu_+
 & \frac{1+r}{2}\mu_-\\[2pt]
 \mathsf{op}=-1
 & \frac{1-r}{2}\mu_-
 & \mu_--\frac{1-r}{2}\mu_+
\end{array}
\]

For insertion, the first row says that a pass returns a
$\mu_+$-sample with mass $(1-r)/2$ and otherwise restarts with
$\mu_-$. Summing this geometric series gives output law $\mu_+$ and
almost-sure termination. Since $\mu_-$ has full support on
$\Sol_{\+C\setminus\{c_0\}}$, termination also holds from every valid
incremental input.

The decremental row is auxiliary, since the first decremental pass
enters with $\mu_+$. By linearity,
\[
 \frac{2}{1-r}\mu_-
 =\mu_+
 +\frac{2}{1-r}
 \left(\mu_--\frac{1-r}{2}\mu_+\right)
\]
says that $2\mu_-/(1-r)$ equals the initial measure plus its one-pass
continuation measure. Each decremental pass returns its entering
assignment with probability $(1-r)/2$, independently of that
assignment. The residual mass therefore tends to zero, and the
assignment-preserving returns have total measure $\mu_-$. In either
direction, the measures entering all outer iterations sum to
$2\mu_-/(1-r)$.

\emph{Time and oracle costs.}
Use the preprocessing in \Cref{remark:efficient-implementation}, and
recall $p_c$ and $\Gamma=(D+1)^2$ from
\Cref{sec:randomized-counting}. The measure calculation above gives
each tree $T$ total mass $w_T/(1-r)$ in either recursive loop. Since
$\iota^T\leq\Gamma|T|$ and the final tests of $\prop$ only reject
tentative trees, \eqref{eq:samp-expected-tree-size} bounds the
expected number of direct recursive updates by
\[
 \frac1{1-r}
 \sum_{T\in\mathfrak T_{\+C,c_0}}\iota^Tw_T
 \le\frac{2p\Gamma}{(1-1.2p)(1-2pD^2)}<\frac14.
\]
The same direct-call bound holds at every depth: grouping calls by constraint family
and sign preserves proportional input laws. Thus there are at most
$4^{-j}$ calls at depth $j$ in expectation.

There are $2/(1-r)=O(1)$ expected outer iterations, so
proposals and tests of $c_0$ have expected oracle cost $O(1)$ and
expected running time $\widetilde O(k\Gamma)$. For a positive-weight tree, a local-draw
attempt has expected oracle cost
$O(\sum_{c\in T}p_c^{-1}+\Gamma|T|)$ and succeeds with probability
$\Pr{\+E_T}/\prod_{c\in T}p_c$.  Multiplying its expected cost until
success by $w_T/(1-r)$ cancels $\Pr{\+E_T}$. The tree-count and
reciprocal-factor bounds in the proof of
\Cref{lem:samp-tree-law} give expected oracle cost $O(1)$ and expected
time $\widetilde O(k)$.

Keep $X$ in one array and restore family membership on recursive
return. Routing and the remaining local work have oracle cost
$O(\Gamma|T|)$ and expected running time
$\widetilde O(k\Gamma^2|T|^2)$ at each selected or resulting tree.
The moment bounds \eqref{eq:samp-expected-tree-size} and
\eqref{eq:samp-tree-size-second-moment}, including singleton mass
$O(p)$, give total expected oracle cost $O(1)$ and expected running
time $\widetilde O(k\Gamma)$.

Summing the corresponding local costs over all depths proves the
claimed bounds.
\end{proof}

\subsection{Weighted \texorpdfstring{$2$}{2}-Tree Proposals}
\label{sec:samp-propose}

We construct the procedure $\prop(\+C,c_0)$ and prove \Cref{lem:samp-tree-law}.

\subsubsection{Construction of $\prop$}

Like the randomized marginal estimator in \Cref{sec:randomized-marginal-estimation}, $\prop$ constructs a
canonical $2$-tree from independent local configurations. Its acceptance
tests cancel the local selection factors and supply the reciprocal
factors in $w_T$. Every recursive proposal is on a strict subfamily of
$\+C$. We use the following procedure, constructed and analyzed in
\Cref{sec:samp-reciprocal-test}.

\begin{lemma}[Reciprocal acceptance test]\label{lem:samp-reciprocal-test}
Let $\+O$ return independent bits with
$\pr{\+O()=0}=B\ge\exp(-3/10)$.
The procedure $\mathsf{ReciprocalTest}(\+O)$ returns a bit that equals
$1$ with probability $1/(2B)$, uses at most $3/B$ expected calls
to $\+O$, and takes $O(1)$ expected additional time.
\end{lemma}

\emph{Local configurations.}
For $c$, a local experiment draws $\tau\sim\mathbb P_{\vbl(c)}$ and
redraws $\tau$ independently while it satisfies $c$ and a fresh
unbiased bit is $1$. Experiments are independent even on overlapping
scopes, and every draw is evaluated once. For the initial proposal in
$\prop$, store the outcome as $\sigma_c$, generating $\sigma_{c_0}$
first and each other $\sigma_c$ only when $\tr$ examines $c$.

\emph{Acceptance tests.}
For a nonempty $T$, assign every $d_i^T$ (hence every $b\in R(T)$) to
the lowest-indexed $c\in T$ at distance at most $2$. For $c\in T$, let
$\+O_{T,c}$ return $0$ iff all the following checks pass:
\begin{enumerate}
\item With probability $1/2$, run a fresh local experiment for $c$ and
require satisfaction.
\item For every assigned $b\in R(T)$, run a fresh local experiment and
require satisfaction.
\item For each assigned $d_i^T$, decreasing in $i$, draw
$T'\gets\prop(\+D_i^T,d_i^T)$; if nonempty, draw an independent copy
$T''$ and require exactly one of
$T',T''$ to be nonempty and even.
\end{enumerate}
Each invocation of $\+O_{T,c}$ uses fresh randomness, independent of
the initial proposal and other invocations.

\begin{algorithm}
\caption{$\prop(\+C,c_0)$.}
\label{alg:samp-propose}
\SetKwInOut{Input}{Input}\SetKwInOut{Output}{Output}
\Input{A constraint set $\+C$ and a constraint $c_0\in\+C$.}
\Output{$T=\varnothing$ or $T\in\mathfrak T_{\+C,c_0}$.}
Generate independent initial configurations $\sigma_c$ by the local
experiment when first needed\;
\lIf{$\sigma_{c_0}$ satisfies $c_0$}{\Return $\varnothing$}
$T\gets\tr_{\+C,c_0}(\{c\in\+C:\sigma_c\text{ violates }c\})$
\label{Line:samp-propose-tree}\;
Extend $(\sigma_c)_{c\in T}$ to $\sigma$ on
$\vbl(T\cup\frin(T))$ with independent uniform values\;
\lIf{$\sigma$ violates a constraint in $\frin(T)$
\label{Line:samp-propose-local}}{\Return $\varnothing$}
\ForEach{$c\in T$}{
 \lIf{$\mathsf{ReciprocalTest}(\+O_{T,c})=0$
  \label{Line:samp-propose-test}}{\Return $\varnothing$}
}
\For{$i=\iota^T,\ldots,1$}{
 Draw $T'\gets\prop(\+D_i^T,d_i^T)$\;
 \lIf{$T'\ne\varnothing$ and $|T'|$ is even
  \label{Line:samp-propose-recursive}}{\Return $\varnothing$}
}
\Return $T$\;
\end{algorithm}

\begin{proof}[Proof of \Cref{lem:samp-tree-law}]
Let $\Gamma=(D+1)^2$ and $H=(1-1.2p)^{-\Gamma}$.
\Cref{cond:sym-counting-LLL} gives $p\Gamma\le1/(4\e)$, $p<1/40$,
$H<9/8$, and $\e pD^2H<9/32$.
By \Cref{lem:marginal-bound},
scope disjointness in $T$, and $\iota^T\le\Gamma|T|$, we have
\[
 w_T\le(pH)^{|T|}.
\]
The rooted-tree estimate in \Cref{lem:rooted-tree-coefficients} then
bounds the total $w_T$-weight of either parity by $1.2p$. Hence
\[
 \sum_{T\in\mathfrak T_{\+C,c_0}}w_T\le2.4p<1.
\]

\emph{Correctness.}
We induct on $|\+C|$; every recursive proposal uses a strict
subfamily. For a recursive input $(\+D,c)$, induction and
\eqref{eq:samp-even-odd-balance} give
\[
 \pr{\text{the two-proposal check passes}}
 =(1-r_{\+D,c})\,
 \pr{\prop(\+D,c)\text{ returns }\varnothing
 \text{ or an odd tree}}.
\]
The equality follows by splitting on whether the first proposal is
empty, even, or odd; the parity-balance identity says the total
odd-tree weight exceeds the total even-tree weight by
$r_{\+D,c}$. This probability is at least one minus the total odd-tree weight,
hence at least $1-1.2p$. The final check has the second probability
above, so relative to the two-proposal check it supplies the factor
$(1-r_{\+D,c})^{-1}$.

For a fixed nonempty $T$, the probability of generating $T$ and
passing Line~\ref{Line:samp-propose-local} is
\[
 \Pr{\+E_T}\prod_{c\in T}\frac2{1+p_c}
 \prod_{b\in R(T)}\frac{1-p_b}{1+p_b}.
\]
For $c\in T$, let $B_c=\pr{\+O_{T,c}()=0}$. Its first check passes
with probability $1/(1+p_c)$. Since at most $D^2$ rejected-constraint
checks and $\Gamma$ recursive checks are assigned to $c$,
\[
 -\log B_c\le3.2p\Gamma<0.3.
\]
Thus \Cref{lem:samp-reciprocal-test} applies. Its acceptance factor
$(2B_c)^{-1}$ cancels the local factors above, while the ratio of the
final and two-proposal checks supplies
$(1-r_{\+D_i^T,d_i^T})^{-1}$. By independence,
$\prop(\+C,c_0)$ returns $T$ with probability $w_T$. The induction
hypothesis and the reciprocal-test guarantee also give almost-sure
termination.

\emph{Efficiency.}
Use the shared preprocessing in \Cref{remark:efficient-implementation}.
Let $T$ be the tentative tree before the remaining tests, with
$T=\varnothing$ when the root is not violated. Each local experiment
uses at most $2$ expected draws and violates its constraint with
probability at most $2p$. A nonroot constraint is examined before its
independent configuration is drawn, and at most $D^2|T|$ constraints
are examined. Therefore,
\begin{equation}\label{eq:samp-expected-tree-size}
 \oE[|T|]\le\frac{2p}{1-2pD^2}.
\end{equation}

For a reciprocal test at $c$, each recursive check is requested before
its fresh randomness is drawn. The $3/B_c$ bound cancels its preceding
pass probabilities, leaving at most
\[
 3\sum_{i=1}^{\Gamma}(1-1.2p)^{-i}\le3\Gamma H
\]
expected recursive checks in that test. Thus all reciprocal tests make
at most $6p\Gamma H/(1-2pD^2)$ expected recursive checks. 
A fixed $T$ reaches the final checks with probability at most
$(pH)^{|T|}$. The same tree-count bound gives
\[
 \sum_{t\ge1}t\upsilon_t(pD^2H)^{t-1}<\frac75,
\]
so the final checks contribute at most $(7/5)p\Gamma H$ in expectation.
Each check makes one proposal and a second only when the first is
nonempty, for at most $1+2.4p$ direct recursive calls in expectation.
Consequently, their
expected number is at most
\[
 (1+2.4p)\left(
 \frac{6p\Gamma H}{1-2pD^2}+\frac75p\Gamma H\right)<0.97.
\]

Apart from recursive proposals, tree generation and acceptance tests
have expected oracle cost
$O(1+\Gamma\oE[|T|])=O(1)$. Conditional on $T$, there are
$O(\Gamma|T|)$ expected direct child calls, each requiring
$O(\Gamma|T|)$ membership-indicator changes. Since every size-$t$
tree is generated with probability at most $(2p)^t$, the rooted-tree
bound and $2\e pD^2\le1/2$ give
\begin{equation}\label{eq:samp-tree-size-second-moment}
 \oE[|T|^2]
 \le2p\sum_{t\ge1}t^2(2\e pD^2)^{t-1}
 \le24p.
\end{equation}
Thus family preparation costs $O(p\Gamma^2)=O(\Gamma)$ in expectation,
and generating and storing local configurations takes
$\widetilde O(k\Gamma)$ expected time. Since each child is requested
before its fresh randomness is exposed, the uniform direct-call bound
sums these costs over all recursion depths. As $\Gamma=O(D^2)$, this
proves both claims.
\end{proof}

\subsubsection{Reciprocal Acceptance Test}
\label{sec:samp-reciprocal-test}

The construction in \Cref{alg:samp-reciprocal-test} is a reciprocal
Bernoulli factory based on the ratio construction of Nacu and
Peres~\cite[Proposition~18]{Nacu2005FastSO}. Its analysis establishes
\Cref{lem:samp-reciprocal-test}, including the bound on the expected
number of calls to the input procedure $\+O$.

\begin{algorithm}
\caption{$\mathsf{ReciprocalTest}(\+O)$.}
\label{alg:samp-reciprocal-test}
\SetKwInOut{Input}{Input}\SetKwInOut{Output}{Output}
\Input{A procedure $\+O$ producing independent bits
with $\pr{\+O()=0}=B\ge\exp(-3/10)$.}
\Output{A bit equal to $1$ with probability $1/(2B)$.}
\SetKwProg{Fn}{Local procedure}{:}{}
\Fn{$\mathsf{Coin}()$}{
 \While{$\+O()=0$}{
  \lIf{a fresh unbiased bit is $1$}{\Return $0$}
 }
 \Return $1$\;
}
\lIf{a fresh unbiased bit is $1$
 \label{Line:samp-reciprocal-first}}{\Return $1$}
$j\gets1+\mathsf{Coin}()$; $i\gets2j-1$
 \label{Line:samp-reciprocal-initial}\;
\lIf{$j=2$ and a fresh unbiased bit is $1$
 \label{Line:samp-reciprocal-second}}{\Return $1$}
\While{$i>0$}{
 \eIf{$\mathsf{Coin}()=0$}{
  $i\gets i-1$\;
 }{
  \If{$\mathsf{Coin}()=1$}{
   $i\gets i+2$; $j\gets j+1$\;
   \lIf{a fresh uniform draw from $\{1,\ldots,j\}$ equals $1$
    \label{Line:samp-reciprocal-later}}{\Return $1$}
  }
 }
}
\Return $0$\;
\end{algorithm}

\begin{proof}[Proof of \Cref{lem:samp-reciprocal-test}]
If $B=1$, all claims are immediate. Assume $B<1$. A call to
$\mathsf{Coin}$ returns $1$ with probability
\(u=2(1-B)/(2-B)<33/80\). For the analysis, suppress the early
returns and contract iterations that leave $(i,j)$ unchanged. The
resulting walk decreases $i$ by $1$ with probability $1-\theta$ and
increases $(i,j)$ by $(2,1)$ with probability
\(\theta=u^2/(1-u+u^2)<9/40\). Its expected number of changes from
level $i$ is at most $i/(1-3\theta)$, so it reaches $i=0$ almost
surely.

\emph{Correctness.}
Let $G(z)$ be the generating function of the number of increases in
the contracted walk started from $i=1$. A path from $i=3$ to $0$
splits into three independent paths from $i=1$ to $0$. Hence
\(G(z)=1-\theta+\theta zG(z)^3\). If $N$ is the final value of $j$,
the initial call gives
\(\oE[z^{N-1}]=(1-u)G(z)+uzG(z)^3\). Solving for $z$ in the equation
for $G$ and changing variables gives
\(\int_0^1G(z)\,dz=1-\theta/(2(1-\theta))\).

Conditional on $N$, the original execution returns $0$ with
probability
\(\frac12\prod_{j=2}^N(1-1/j)=1/(2N)\). Integrating the preceding
identity for \(\oE[z^{N-1}]\), and using
\(B=2(1-u)/(2-u)\), shows that
$\mathsf{ReciprocalTest}(\+O)$ returns $1$ with probability $1/(2B)$.

\emph{Calls to $\+O$.}
Let $M$ be the expected sum of $1/j$ immediately before counter
changes in the contracted walk started at $(3,2)$. Starting from
$(1,1)$, the corresponding expectation is $1+\theta M$. Summing
$1/j$ over upward changes gives
\[
 1+\theta M
 =\frac1\theta\int_0^1\frac{1-G(z)}{1-z}\,dz
 =\int_0^1\frac{dz}{1-\theta z(1+G(z)+G(z)^2)}
 \le\frac{-\log(1-3\theta)}{3\theta}.
\]
Since $G(z)\le1$ and $\theta<9/40$, we have $M<74/25$.

Counting the initial call, self-loops, and return tests, and using
that $\mathsf{Coin}$ makes $2/(2-B)$ expected calls to $\+O$, gives
\[
 B\,\oE[\text{calls to }\+O]
 =1-u+\frac{1-u^2}{1-u+u^2}
 \left(1-u+\frac{u}{1-u+u^2}M\right)<3.
\]
Indeed, after replacing $M$ by $74/25$, the right-hand side is
increasing on $u\in[0,33/80]$ and is less than $3$ at $u=33/80$.
Thus the expected number of calls to $\+O$ is at most $3/B$.

\emph{Computational overhead.}
The generating-function identities give $\oE[N^2]=O(1)$. Before an
early return, there are at most $3N$ counter changes; restoring
self-loops gives at most $3N/(1-u+u^2)\le4N$ expected iterations.
Throughout the execution, $i\le2j-1\le2N-1$, so counter operations and
uniform draws take $O(\log(N+1))$ bit operations each. Their expected
total cost is $O(\oE[N\log(N+1)])=O(1)$. The fair bits inside
$\mathsf{Coin}$ are charged to its calls to $\+O$.
\end{proof}

\subsection{Routing Between \texorpdfstring{$2$}{2}-Tree Parities}
\label{sec:samp-route}
We construct the procedure $\mathsf{Route}(T,X)$ and prove
\Cref{lem:samp-parity-pairing}. 
The construction uses the execution of $\tr_{\+C,c_0}(T)$ in \Cref{Alg:gen-tree}.

For a pair $(T,X)$ in the support of either measure in
\Cref{lem:samp-parity-pairing}, run $\tr_{\+C,c_0}(T)$.
This execution selects exactly the constraints in $T\cup R(T)$, in a
fixed order.
\[
 \+A(X,T)\defeq
 \left\{c\in(T\cup R(T))\setminus\{c_0\}:
 \begin{array}{l}
 X\text{ violates }c,\text{ and }\dist_{G(\+C)}(c,d)>2\\
 \text{for every later selected constraint }d\text{ violated by }X
 \end{array}\right\}.
\]
Thus, $\+A(X,T)$ consists of the nonroot violated constraints that can
be added to or removed from $T$ without changing the later selected
constraints that $X$ violates.

\begin{algorithm}
\caption{$\mathsf{Route}(T,X)$.}
\label{alg:samp-route}
\SetKwInOut{Input}{Input}\SetKwInOut{Output}{Output}
\Input{A pair $(T,X)$ in the support of either
$\nu_{\mathrm{even}}$ or $\nu_{\mathrm{odd}}$.}
\Output{A pair $(T',X)$ in the support of the other measure.}
Run $\tr_{\+C,c_0}(T)$, recording each chosen constraint and whether
$X$ violates it\label{Line:samp-route-trace}\;
$\+A\gets\+A(X,T)$\label{Line:samp-route-safe}\;
\lIf{$\+A=\varnothing$}{\Return $(T,X)$}
Draw $c$ uniformly from $\+A$\label{Line:samp-route-choose}\;
\Return $(T\triangle\{c\},X)$\label{Line:samp-route-return}\;
\end{algorithm}

\begin{proof}[Proof of \Cref{lem:samp-parity-pairing}]
For every $T$ and $X$ with $w_T\nu_T(X)>0$, the proof of
\Cref{lem:distributional-2-tree-expansion} gives
\[
 w_T\nu_T(X)=\frac{1}{Z(\+C\setminus\{c_0\})}.
\]

\emph{Correctness.}
If $T=\{c_0\}$ and
$X\in\Sol_{\+C\setminus\{c_0\}}$ violates $c_0$, the exceptional
singleton term in $\nu_{\mathrm{even}}$ and the root term in
$\nu_{\mathrm{odd}}$ both have this mass.
No nonroot selected
constraint is then violated by $X$, so $\+A(X,T)=\varnothing$ and
$\mathsf{Route}$ fixes the pair.

Consider every other positive-mass pair. Then
$w_T\nu_T(X)>0$, so $X$ violates every constraint in $T$ and
satisfies both $\frin(T)$ and $\+D_0^T$. Since
$\+D_0^T=\fr(T)\setminus\frin(T)$, it satisfies every constraint in
$\fr(T)$. Hence every constraint violated by $X$ lies in $T\cup R(T)$,
and at least one nonroot constraint is violated.

Fix $c\in\+A(X,T)$ and let $T'=T\triangle\{c\}$. The executions with
inputs $T$ and $T'$ agree until $c$ is selected, where acceptance and
rejection are interchanged. Every later selected constraint violated
by $X$ is at distance greater than two from $c$, so it receives the
same decision in both executions. Any other affected later query is
satisfied by $X$ and is rejected. Thus the ordered list of selected
constraints violated by $X$ is unchanged, and the execution with
input $T'$ returns $T'$.

By \Cref{lem:property-2-tree},
$T'\in\mathfrak T_{\+C,c_0}$. Moreover, $X$ violates $T'$ and
satisfies $\fr(T')$. Since $\+A(X,T)$ depends only on the ordered
selected constraints violated by $X$,
\[
 \+A(X,T')=\+A(X,T)
 \qquad\text{and}\qquad
 w_{T'}\nu_{T'}(X)=w_T\nu_T(X).
\]
Each nontrivial move changes the parity of $T$ and has the same
probability as its reverse. Together with the fixed singleton mass,
summing incoming contributions proves both measure-transfer statements.

\emph{Efficiency.}
The execution in Line~\ref{Line:samp-route-trace} selects
$O(D^2|T|)$ constraints, so recording which are violated has oracle
cost $O(D^2|T|)$. To compute $\+A(X,T)$, scan the recorded violated
constraints backward and mark the distance-two neighborhood of each
later one. Each marking uses $O(D^2)$ graph operations. Thus
$\mathsf{Route}$ has oracle cost $O(D^2|T|)$ and expected running time
$\widetilde O(kD^4|T|)$.
\end{proof}

\section{Conclusions and Open Problems}\label{sec:conclusions}

We give efficient approximate counting and exact sampling algorithms for general CSPs under $4\e p(D+1)^2\leq1$, matching the known hardness threshold up to constant factors.

Our canonical $2$-tree expansion reorganizes inclusion--exclusion for constraint marginals, exposing the $pD^2$ scale and yielding exponential correlation decay throughout this regime.

Truncating this recursion gives deterministic marginal estimates with exponentially small error in the truncation budget. Sampling a canonical $2$-tree and recursively estimating its reciprocal factors gives unbiased estimates with bounded second moment and subcritical expected work. Constraint-wise self-reducibility then yields the counting algorithms in \Cref{thm:fptas-counting-LLL} and \Cref{thm:fpras-counting-LLL}.

For sampling, we apply the $2$-tree expansion to the unnormalized
constraint-violation measure.  Separating its terms by parity gives an
identity balancing the violation and even-tree terms against the
odd-tree terms.  We turn this identity into an exact
rejection-and-repair dynamic update, using weighted $2$-tree proposals
and recursive constraint removals and insertions to repair a failed
insertion back to the previous conditional law.  Constraint-wise
self-reducibility then composes these updates into the near-linear-time
exact sampler of \Cref{thm:fpaus-sampling-LLL}.

Our results leave three questions:
\begin{itemize}
\item What are the precise approximate counting and sampling thresholds
for general CSPs?  In particular, what is the optimal constant in the
symmetric LLL condition, and is there a sharp asymmetric tractability
region analogous to Shearer's region for satisfiability?

\item Is there a simple local dynamics or absorbing process for sampling
systems with higher-order interactions up to their critical
tractability threshold?  The target distribution might arise either at
equilibrium or upon absorption from a suitable initialization.  Such a
framework could unify Glauber dynamics for monotone CNFs and the
Moser--Tardos algorithm for extremal CSPs as opposite limiting cases.

\item Can the canonical $2$-tree expansion extend beyond independent
variables to an abstract LLL framework based on resampling
oracles~\cite{harvey2020algorithmic}?  More broadly, which
problem-specific combinatorial structures support analogous marginal
expansions and yield sharp counting and sampling algorithms?
\end{itemize}

\section*{AI Methodology}

The authors used ChatGPT 5.5, 5.6, and 6 as AI assistants during the preparation of this manuscript.
Early in this work, AI suggested a marginal estimator based on a branching process, combining an inclusion--exclusion decomposition of marginal probabilities with classical Monte Carlo techniques such as Russian roulette estimation.
This initial construction required stronger local lemma conditions and stronger oracle assumptions.

Inspired by this initial construction, the authors developed {the main counting results}, including the $2$-tree expansion, the analysis of its correlation decay properties, and the deterministic approximate counting algorithm based on truncated 2-tree expansions. 
The authors also redesigned and simplified the randomized approximate counting algorithm, achieving the stated local lemma condition under the weaker evaluation-oracle assumption.

AI assistance also contributed to the development of the sampling
algorithms through extensive interactions involving several authors.
The authors contributed many original ideas throughout this process,
which progressed from a polynomial-time approximate sampler with a
parameter-dependent exponent, through cubic- and quadratic-time
approximate samplers, to the near-linear exact sampler presented here.

The authors devoted substantial effort to reorganizing and
rewriting the manuscript, including all key arguments and
algorithmic descriptions, to make the ideas and proofs clear
and accessible to researchers.
They also verified all mathematical statements, proofs,
algorithms, and references, and take full responsibility
for all content.

\ifdoubleblind
\else
\section*{Acknowledgments}
We thank Tianxing Ding and Yixiao Yu for helpful discussions.
\fi

\bibliographystyle{alpha}
\bibliography{references}

@String{Computing = "Computing" }

@String{Computer = "{IEEE} Computer" }

@String{Springer = "Springer-Verlag" }

@inproceedings{achlioptas2003threshold,
  author    = {Achlioptas, Dimitris and Peres, Yuval},
  title     = {The threshold for random $k$-{SAT} is $2^k \ln 2 - {O}(k)$},
  year      = {2003},
  isbn      = {1581136749},
  publisher = {ACM},
  url       = {https://doi.org/10.1145/780542.780577},
  doi       = {10.1145/780542.780577},
  booktitle = {STOC},
  pages     = {223--231},
  numpages  = {9}
}

@article{achlioptas2016random,
  author    = {Achlioptas, Dimitris and Iliopoulos, Fotis},
  title     = {Random Walks That Find Perfect Objects and the {L}ov\'{a}sz Local Lemma},
  year      = {2016},
  publisher = {ACM},
  volume    = {63},
  number    = {3},
  issn      = {0004-5411},
  url       = {https://doi.org/10.1145/2818352},
  doi       = {10.1145/2818352},
  journal   = {J. ACM},
  articleno = {22},
  numpages  = {29},
  note      = {(Conference version in \emph{FOCS}'14)}
}

@inproceedings{achlioptasAsymptoticOrderRandom2002,
  title     = {The asymptotic order of the random $k$-{SAT} threshold},
  booktitle = {FOCS},
  author = {Achlioptas, Dimitris and Moore, Cristopher},
  year      = 2002,
  pages     = {779--788},
  doi       = {10.1109/SFCS.2002.1182003}
}

@article{alon1991parallel,
  author  = {Alon, Noga},
  journal = {Random Struct. Algorithms},
  number  = {4},
  pages   = {367--378},
  title   = {A parallel algorithmic version of the local lemma},
  volume  = {2},
  year    = {1991},
  doi     = {10.1002/rsa.3240020403},
  note    = {(Conference version in \emph{FOCS}'91)}
}

@article{barvinok2026computing,
  author  = {Barvinok, Alexander},
  title   = {Computing the probability of intersection},
  journal = {Combin. Probab. Comput.},
  year    = {2026},
  volume  = {},
  number  = {},
  pages   = {1--20},
  doi     = {10.1017/S0963548326100492}
}

@article{beck1991algorithmic,
  author  = {Beck, J{\'o}zsef},
  journal = {Random Struct. Algorithms},
  number  = {4},
  pages   = {343--365},
  title   = {An algorithmic approach to the {L}ov{\'a}sz local lemma.},
  doi     = {10.1002/rsa.3240020402},
  volume  = {2},
  year    = {1991}
}

@article{BGGGS19,
  author  = {Ivona Bez{\'{a}}kov{\'{a}} and Andreas Galanis and Leslie A. Goldberg and Heng Guo and Daniel {\v{S}}tefankovi{\v{c}}},
  journal = {{SIAM} J. Comput.},
  number  = {2},
  pages   = {279--349},
  title   = {Approximation via Correlation Decay When Strong Spatial Mixing Fails},
  volume  = {48},
  doi     = {10.1137/16M1083906},
  year    = {2019},
  note    = {(Conference version in \emph{ICALP}'16)}
}

@inproceedings{blanchet2015unbiased,
  author       = {Blanchet, Jose and Chen, Nan and Glynn, Peter W.},
  title        = {Unbiased {M}onte {C}arlo Computation of Smooth Functions of Expectations via {T}aylor Expansions},
  booktitle    = {Proceedings of the 2015 Winter Simulation Conference},
  pages        = {360--367},
  year         = {2015},
  doi          = {10.1109/WSC.2015.7408178},
  organization = {IEEE}
}

@article{borgs2013left,
  author     = {Borgs, Christian and Chayes, Jennifer and Kahn, Jeff and
                {Lov\'{a}sz}, {L\'{a}szl\'{o}}},
  title      = {Left and right convergence of graphs with bounded degree},
  journal    = {Random Struct. Algorithms},
  fjournal   = {Random Struct. \& Algorithms},
  volume     = {42},
  year       = {2013},
  number     = {1},
  pages      = {1--28},
  issn       = {1042-9832},
  mrclass    = {05C80 (60C05)},
  mrnumber   = {2999210},
  mrreviewer = {Christian Lavault},
  doi        = {10.1002/rsa.20414},
  url        = {https://doi.org/10.1002/rsa.20414}
}

@article{borgs2022potts,
  author   = {Borgs, Christian and Chayes, Jennifer and Helmuth, Tyler and Perkins, Will and Tetali, Prasad},
  title    = {Efficient sampling and counting algorithms for the {P}otts model on {$\mathbb{Z}^d$} at all temperatures},
  journal  = {Random Struct. Algorithms},
  volume   = {63},
  number   = {1},
  pages    = {130-170},
  doi      = {https://doi.org/10.1002/rsa.21131},
  url      = {https://onlinelibrary.wiley.com/doi/abs/10.1002/rsa.21131},
  eprint   = {https://onlinelibrary.wiley.com/doi/pdf/10.1002/rsa.21131},
  year     = {2023},
  note     = {(Conference version in \emph{STOC}'20)}
}

@inproceedings{bulatov2017dichotomy,
  author    = {Andrei A. Bulatov},
  title     = {A Dichotomy Theorem for Nonuniform {CSPs}},
  booktitle = {FOCS},
  pages     = {319--330},
  year      = {2017},
  doi       = {10.1109/FOCS.2017.37},
  url       = {https://doi.org},
  publisher = {IEEE}
}

@inproceedings{cannon2020counting,
  author    = {Cannon, Sarah and Perkins, Will},
  title     = {Counting independent sets in unbalanced bipartite graphs},
  booktitle = {SODA},
  pages     = {1456--1466},
  publisher = {SIAM},
  year      = {2020},
  mrclass   = {68W25},
  mrnumber  = {4141270}
}

@article{chandrasekaran2013deterministic,
  author  = {Chandrasekaran, Karthekeyan and Goyal, Navin and Haeupler, Bernhard},
  title   = {Deterministic Algorithms for the {L}ovász Local Lemma},
  journal = {SIAM J. Comput.},
  volume  = {42},
  number  = {6},
  pages   = {2132-2155},
  year    = {2013},
  doi     = {10.1137/100799642},
  url     = { 
             https://doi.org/10.1137/100799642
             },
  eprint  = { 
             https://doi.org/10.1137/100799642
             },
  note    = {(Conference version in \emph{SODA}'10)}
}

@article{chen2024fast,
  author  = {Chen, Zongchen and Galanis, Andreas and Goldberg, Leslie Ann and Guo, Heng and Herrera-Poyatos, Andr{\'e}s and Mani, Nitya and Moitra, Ankur},
  title   = {Fast Sampling of Satisfying Assignments from Random {$k$-SAT} with Applications to Connectivity},
  journal = {SIAM J. Discrete Math.},
  volume  = {38},
  number  = {4},
  pages   = {2750-2811},
  year    = {2024},
  doi     = {10.1137/23M1595722},
  url     = { 
             
             https://doi.org/10.1137/23M1595722
             
             
             
             },
  eprint  = { 
             
             https://doi.org/10.1137/23M1595722
             
             
             
             }
}

@inproceedings{chen2025counting,
  author    = {Chen, Zongchen and Lonkar, Aditya and Wang, Chunyang and Yang, Kuan and Yin, Yitong},
  title     = {Counting Random {$k$-SAT} near the Satisfiability Threshold},
  year      = {2025},
  isbn      = {9798400715105},
  publisher = {ACM},
  url       = {https://doi.org/10.1145/3717823.3718163},
  doi       = {10.1145/3717823.3718163},
  booktitle = {STOC},
  pages     = {867--878},
  numpages  = {12}
}

@inproceedings{coja2014asymptotic,
  author    = {{Coja-Oghlan}, Amin},
  title     = {The asymptotic $k$-{SAT} threshold},
  year      = {2014},
  isbn      = {9781450327107},
  publisher = {ACM},
  url       = {https://doi.org/10.1145/2591796.2591822},
  doi       = {10.1145/2591796.2591822},
  booktitle = {STOC},
  pages     = {804--813},
  numpages  = {10}
}

@inproceedings{cook1971complexity,
  author    = {Cook, Stephen A.},
  title     = {The complexity of theorem-proving procedures},
  year      = {1971},
  isbn      = {9781450374644},
  publisher = {ACM},
  url       = {https://doi.org/10.1145/800157.805047},
  doi       = {10.1145/800157.805047},
  booktitle = {STOC},
  pages     = {151--158},
  numpages  = {8}
}

@article{CS00,
  author  = {Czumaj, Artur and Scheideler, Christian},
  title   = {Coloring nonuniform hypergraphs: a new algorithmic approach to
             the general {L}ov\'{a}sz local lemma},
  journal = {Random Struct. Algorithms},
  volume  = {17},
  year    = {2000},
  number  = {3--4},
  pages   = {213--237},
  doi     = {10.1002/1098-2418(200010/12)17:3/4<213::AID-RSA3>3.0.CO;2-Y},
  note    = {(Conference version in \emph{SODA}'00)}
}

@article{ding2022satisfiability,
  author    = {Jian Ding and Allan Sly and Nike Sun},
  title     = {Proof of the satisfiability conjecture for large $k$},
  volume    = {196},
  journal   = {Ann. Math.},
  number    = {1},
  publisher = {Department of Mathematics of Princeton University},
  pages     = {1--388},
  year      = {2022},
  doi       = {10.4007/annals.2022.196.1.1},
  url       = {https://doi.org/10.4007/annals.2022.196.1.1}
}

@article{feder1998computational,
  author  = {Tom{\'{a}}s Feder and Moshe Y. Vardi},
  title   = {The Computational Structure of Monotone Monadic {SNP} and Constraint Satisfaction: {A} Study through Datalog and Group Theory},
  journal = {{SIAM} J. Comput.},
  volume  = {28},
  number  = {1},
  pages   = {57--104},
  year    = {1998},
  doi     = {10.1137/S0097539794266766},
  note    = {(Conference version in \emph{STOC}'93)}
}

@inproceedings{feng2021sampling,
  title     = {Sampling constraint satisfaction solutions in the local lemma regime},
  author    = {Feng, Weiming and He, Kun and Yin, Yitong},
  booktitle = {STOC},
  pages     = {1565--1578},
  publisher = {ACM},
  doi       = {10.1145/3406325.3451101},
  year      = {2021}
}

@inproceedings{feng2022improved,
  author    = {Weiming Feng and Heng Guo and Jiaheng Wang},
  booktitle = {RANDOM},
  pages     = {25:1--25:17},
  publisher = {Schloss Dagstuhl -- Leibniz-Zentrum f\"ur Informatik},
  title     = {Improved Bounds for Randomly colouring Simple Hypergraphs},
  doi       = {10.4230/LIPIcs.APPROX/RANDOM.2022.25},
  year      = {2022}
}

@article{feng2025toward,
  author  = {Feng, Weiming and Guo, Heng and Wang, Chunyang and Wang, Jiaheng and Yin, Yitong},
  title   = {Toward Derandomizing {M}arkov Chain {M}onte {C}arlo},
  journal = {SIAM J. Comput.},
  volume  = {54},
  number  = {3},
  pages   = {775--813},
  year    = {2025},
  doi     = {10.1137/24M1663806},
  url     = {https://epubs.siam.org/doi/abs/10.1137/24M1663806},
  note    = {(Conference version in \emph{FOCS}'23)}
}

@article{FGYZ20,
  author   = {Feng, Weiming and Guo, Heng and Yin, Yitong and Zhang, Chihao},
  fjournal = {Journal of the ACM},
  journal  = {J. ACM},
  number   = {6},
  pages    = {Art. 40, 42},
  title    = {Fast sampling and counting {$k$}-{SAT} solutions in the local lemma regime},
  volume   = {68},
  doi      = {10.1145/3469832},
  note     = {(Conference version in \emph{STOC}'20)},
  year     = {2021}
}

@article{friedgut1999sharp,
  issn      = {08940347, 10886834},
  url       = {http://www.jstor.org/stable/2646096},
  author    = {Ehud Friedgut and Jean Bourgain},
  journal   = {J. Am. Math. Soc.},
  number    = {4},
  pages     = {1017--1054},
  publisher = {American Mathematical Society},
  title     = {Sharp Thresholds of Graph Properties, and the $k$-{SAT} Problem},
  urldate   = {2024-10-13},
  volume    = {12},
  year      = {1999}
}

@article{galanis2016inapproximability,
  author  = {Galanis, Andreas and {\v{S}}tefankovi{\v{c}}, Daniel and Vigoda, Eric},
  journal = {Comb. Probab. Comput.},
  number  = {4},
  pages   = {500--559},
  title   = {Inapproximability of the partition function for the antiferromagnetic {I}sing and hard-core models},
  volume  = {25},
  year    = {2016}
}

@article{galanis2023inapproximability,
  author     = {Galanis, Andreas and Guo, Heng and Wang, Jiaheng},
  title      = {Inapproximability of Counting Hypergraph colorings},
  year       = {2023},
  issue_date = {December 2022},
  publisher  = {ACM},
  address    = {New York, NY, USA},
  volume     = {14},
  number     = {3–4},
  issn       = {1942-3454},
  url        = {https://doi.org/10.1145/3558554},
  doi        = {10.1145/3558554},
  journal    = {ACM Trans. Comput. Theory},
  articleno  = {10},
  pages      = {1--33}
}

@article{GGGY21,
  title     = {Counting solutions to random {SAT} formulas},
  author    = {Galanis, Andreas and Goldberg, Leslie Ann and Guo, Heng and Yang, Kuan},
  journal   = {SIAM J. Comput.},
  volume    = {50},
  number    = {6},
  pages     = {1701--1738},
  year      = {2021},
  publisher = {SIAM},
  note      = {(Conference version in \emph{ICALP}'20)}
}

@article{GJL19,
  author   = {Guo, Heng and Jerrum, Mark and Liu, Jingcheng},
  title    = {Uniform sampling through the {L}ov\'{a}sz local lemma},
  journal  = {J. ACM},
  fjournal = {Journal of the ACM},
  volume   = {66},
  year     = {2019},
  number   = {3},
  pages    = {Art. 18, 31},
  issn     = {0004-5411},
  mrclass  = {68W20 (05D40 60C05)},
  mrnumber = {3941342},
  doi      = {10.1145/3310131},
  note     = {(Conference version in \emph{STOC}'17)}
}

@article{guo2019counting,
  title     = {Counting hypergraph colorings in the local lemma regime},
  author    = {Guo, Heng and Liao, Chao and Lu, Pinyan and Zhang, Chihao},
  journal   = {SIAM J. Comput.},
  volume    = {48},
  number    = {4},
  pages     = {1397--1424},
  year      = {2019},
  doi       = {10.1137/18M1202955},
  note      = {(Conference version in \emph{STOC}'18)},
  publisher = {SIAM}
}

@article{haeupler2011new,
  author    = {Bernhard Haeupler and
               Barna Saha and
               Aravind Srinivasan},
  title     = {New Constructive Aspects of the {L}ov{\'{a}}sz Local Lemma},
  journal   = {J. {ACM}},
  volume    = {58},
  number    = {6},
  pages     = {28:1--28:28},
  year      = {2011},
  doi       = {10.1145/2049697.2049702},
  bibsource = {dblp computer science bibliography, https://dblp.org},
  note      = {(Conference version in \emph{FOCS}'10)}
}

@article{harris2019moser,
  author    = {Harris, David G. and Srinivasan, Aravind},
  title     = {The {Moser--Tardos} Framework with Partial Resampling},
  year      = {2019},
  publisher = {ACM},
  volume    = {66},
  number    = {5},
  issn      = {0004-5411},
  url       = {https://doi.org/10.1145/3342222},
  doi       = {10.1145/3342222},
  journal   = {J. ACM},
  articleno = {36},
  numpages  = {45},
  note      = {(Conference version in \emph{FOCS}'13)}
}

@article{harvey2020algorithmic,
  author  = {Harvey, Nicholas J. A. and Vondr{\'a}k, Jan},
  title   = {An Algorithmic Proof of the {L}ovász Local Lemma via Resampling Oracles},
  journal = {SIAM J. Comput.},
  volume  = {49},
  number  = {2},
  pages   = {394-428},
  year    = {2020},
  doi     = {10.1137/18M1167176},
  url     = { 
             
             https://doi.org/10.1137/18M1167176
             
             
             
             },
  eprint  = { 
             
             https://doi.org/10.1137/18M1167176
             
             
             
             },
  note    = {(Conference version in \emph{FOCS}'15)}
}

@inproceedings{he2022counting,
  author    = {Kun He and Chunyang Wang and Yitong Yin},
  title     = {Deterministic counting {L}ovász local lemma beyond linear programming},
  author+an = {2=thesisauthor},
  booktitle = {SODA},
  publisher = {SIAM},
  chapter   = {},
  pages     = {3388--3425},
  year      = {2023},
  doi       = {10.1137/1.9781611977554.ch130}
}

@inproceedings{he2022sampling,
  author    = {He, Kun and Wang, Chunyang and Yin, Yitong},
  booktitle = {FOCS},
  author+an = {2=thesisauthor},
  publisher = {IEEE},
  title     = {Sampling {L}ovász local lemma for general constraint satisfaction solutions in near-linear time},
  year      = {2022},
  volume    = {},
  number    = {},
  pages     = {147--158},
  doi       = {10.1109/FOCS54457.2022.00021}
}

@inproceedings{he2023improved,
  author    = {Kun He and Kewen Wu and Kuan Yang},
  title     = {Improved Bounds for Sampling Solutions of Random {CNF} Formulas},
  booktitle = {SODA},
  publisher = {SIAM},
  chapter   = {},
  pages     = {3330--3361},
  year      = {2023},
  doi       = {10.1137/1.9781611977554.ch128}
}

@article{helmuth2020algorithmic,
  title     = {Algorithmic {P}irogov-{S}inai theory},
  author    = {Helmuth, Tyler and Perkins, Will and Regts, Guus},
  journal   = {Probability Theory and Related Fields},
  volume    = {176},
  number    = {3--4},
  pages     = {851--895},
  year      = {2020},
  publisher = {Springer},
  note      = {(Conference version in \emph{STOC}'19)}
}

@article{HSW21,
  archiveprefix = {arXiv},
  author        = {Kun He and Xiaoming Sun and Kewen Wu},
  eprint        = {2107.03932},
  journal = {arXiv preprint arXiv:2107.03932},
  title         = {Perfect sampling for (atomic) {L}ov\'{a}sz local lemma},
  eprintclass   = {cs.DS},
  year          = {2021}
}

@article{HSZ19,
  author  = {Jonathan Hermon and Allan Sly and Yumeng Zhang},
  journal = {Random Struct. Algorithms},
  number  = {4},
  pages   = {730--767},
  title   = {Rapid mixing of hypergraph independent sets},
  volume  = {54},
  year    = {2019}
}

@article{jenssen2020algorithms,
  author   = {Jenssen, Matthew and Keevash, Peter and Perkins, Will},
  title    = {Algorithms for \#{BIS}-hard problems on expander graphs},
  journal  = {SIAM J. Comput.},
  fjournal = {SIAM Journal on Computing},
  volume   = {49},
  year     = {2020},
  number   = {4},
  pages    = {681--710},
  issn     = {0097-5397},
  mrclass  = {68Q25 (68Q87 82B20)},
  mrnumber = {4118344},
  doi      = {10.1137/19M1286669},
  url      = {https://doi.org/10.1137/19M1286669},
  note     = {(Conference version in \emph{SODA}'19)}
}

@article{kirousis1998approximate,
  author   = {Kirousis, Lefteris M. and Kranakis, Evangelos and Krizanc, Danny and Stamatiou, Yannis C.},
  title    = {Approximating the unsatisfiability threshold of random formulas},
  journal  = {Random Struct. Algorithms},
  volume   = {12},
  number   = {3},
  pages    = {253-269},
  doi      = {https://doi.org/10.1002/(SICI)1098-2418(199805)12:3<253::AID-RSA3>3.0.CO;2-U},
  year     = {1998}
}

@inproceedings{Kolipaka2011MoserAT,
  author    = {Kashyap Babu Rao Kolipaka and Mario Szegedy},
  booktitle = {STOC},
  pages     = {235--244},
  title     = {Moser and {T}ardos meet {L}ov{\'a}sz},
  doi       = {10.1145/1993636.1993669},
  year      = {2011}
}

@article{kolmogorov2018commutativity,
  author  = {Kolmogorov, Vladimir},
  title   = {Commutativity in the Algorithmic {L}ovász Local Lemma},
  journal = {SIAM J. Comput.},
  volume  = {47},
  number  = {6},
  pages   = {2029--2056},
  year    = {2018},
  doi     = {10.1137/16M1093306},
  url     = { 
             
             https://doi.org/10.1137/16M1093306
             
             
             
             },
  eprint  = { 
             
             https://doi.org/10.1137/16M1093306
             
             
             
             },
  note    = {(Conference version in \emph{FOCS}'16)}
}

@article{Kotecky1986cluster,
  author  = {Koteck{\'y}, R. and Preiss, D.},
  title   = {Cluster expansion for abstract polymer models},
  journal = {Commun. Math. Phys.},
  volume  = {103},
  number  = {3},
  pages   = {491--498},
  year    = {1986}
}

@inproceedings{li2013correlation,
  author    = {Li, Liang and Lu, Pinyan and Yin, Yitong},
  title     = {Correlation Decay up to Uniqueness in Spin Systems},
  year      = {2013},
  isbn      = {9781611972511},
  publisher = {SIAM},
  booktitle = {SODA},
  pages     = {67--84},
  numpages  = {18},
  location  = {New Orleans, Louisiana}
}

@article{LocalLemma,
  author  = {Erd\H{o}s, Paul and Lov\'asz, L\'aszl\'o},
  journal = {Infinite and finite sets, volume 10 of Colloquia Mathematica Societatis J\'anos Bolyai},
  pages   = {609-628},
  title   = {Problems and results on 3-chromatic Hypergraphs and some related questions},
  url     = {https://www.semanticscholar.org/paper/Problems-and-Results-on-3-chromatic-Hypergraphs-and-Lov\%C3\%A1sz/65e463b2e1ec161e3263029f1030f318ef06669d},
  year    = {1975}
}

@article{mann2025approximate,
  title         = {Approximate Counting in Local Lemma Regimes},
  author        = {Ryan L. Mann and Gabriel Waite},
  year          = {2025},
  eprint        = {2512.10134},
  archiveprefix = {arXiv},
  journal       = {arXiv preprint arXiv:2512.10134},
  primaryclass  = {cs.DS},
  url           = {https://arxiv.org/abs/2512.10134}
}

@article{Moi19,
  author  = {Ankur Moitra},
  journal = {J. {ACM}},
  number  = {2},
  pages   = {10:1--10:25},
  title   = {Approximate Counting, the {L}ov{\'{a}}sz Local Lemma, and Inference in Graphical Models},
  volume  = {66},
  year    = {2019},
  doi     = {10.1145/3268930},
  note    = {(Conference version in \emph{STOC}'17)}
}

@inproceedings{molloy1998further,
  author    = {Molloy, Michael and Reed, Bruce},
  booktitle = {STOC},
  pages     = {524--529},
  title     = {Further algorithmic aspects of the local lemma},
  doi       = {10.1145/276698.276866},
  year      = {1998}
}

@inproceedings{moser2009constructive,
  author    = {Moser, Robin A.},
  booktitle = {STOC},
  publisher = {ACM},
  pages     = {343--350},
  title     = {A constructive proof of the {L}ov{\'a}sz local lemma},
  doi       = {10.1145/1536414.1536462},
  year      = {2009}
}

@article{moser2010constructive,
  author        = {Moser, Robin A. and Tardos, G{\'a}bor},
  journal       = {J. {ACM}},
  number        = {2},
  pages         = {11:1-11:15},
  publisher     = {ACM},
  title         = {A constructive proof of the general {L}ov{\'a}sz Local Lemma},
  volume        = {57},
  doi           = {10.1145/1667053.1667060},
  year          = {2010}
}

@article{Nacu2005FastSO,
  author     = {Nacu, \c{S}erban and Peres, Yuval},
  title      = {Fast simulation of new coins from old},
  journal    = {Ann. Appl. Probab.},
  fjournal   = {The Annals of Applied Probability},
  volume     = {15},
  year       = {2005},
  number     = {1A},
  pages      = {93--115},
  issn       = {1050-5164},
  mrclass    = {60C05 (65C50)},
  mrnumber   = {2115037},
  mrreviewer = {Ralph Neininger},
  doi        = {10.1214/105051604000000549}
}

@article{patel2017deterministic,
  author   = {Patel, Viresh and Regts, Guus},
  title    = {Deterministic polynomial-time approximation algorithms for
              partition functions and graph polynomials},
  journal  = {SIAM J. Comput.},
  fjournal = {SIAM Journal on Computing},
  volume   = {46},
  year     = {2017},
  number   = {6},
  pages    = {1893--1919},
  issn     = {0097-5397},
  mrclass  = {68W25 (05C31 05C60 05C69)},
  mrnumber = {3738853},
  doi      = {10.1137/16M1101003},
  url      = {https://doi.org/10.1137/16M1101003}
}

@inproceedings{qiu2022perfect,
  author    = {Guoliang Qiu and
               Yanheng Wang and
               Chihao Zhang},
  title     = {A Perfect Sampler for Hypergraph Independent Sets},
  booktitle = {ICALP},
  pages     = {103:1--103:16},
  publisher = {Schloss Dagstuhl -- Leibniz-Zentrum f\"ur Informatik},
  year      = {2022}
}

@inproceedings{schaefer1978complexity,
  author    = {Schaefer, Thomas J.},
  title     = {The complexity of satisfiability problems},
  booktitle = {STOC},
  year      = {1978},
  pages     = {216--226},
  publisher = {ACM}
}

@article{shearer85,
  author   = {Shearer, James B.},
  title    = {On a problem of {S}pencer},
  journal  = {Combinatorica},
  fjournal = {Combinatorica. An International Journal of the J\'{a}nos Bolyai
              Mathematical Society},
  volume   = {5},
  year     = {1985},
  number   = {3},
  pages    = {241--245},
  issn     = {0209-9683},
  doi      = {10.1007/BF02579368}
}

@article{sinclair2014approximation,
  author   = {{Sinclair}, Alistair and {Srivastava}, Piyush and {Thurley}, Marc},
  title    = {Approximation Algorithms for Two-State Anti-Ferromagnetic Spin Systems on Bounded Degree Graphs},
  journal  = {J. Stat. Phys.},
  year     = 2014,
  volume   = {155},
  number   = {4},
  pages    = {666--686},
  doi      = {10.1007/s10955-014-0947-5},
  adsurl   = {https://ui.adsabs.harvard.edu/abs/2014JSP...155..666S},
  note     = {(Conference version in \emph{SODA}'12)}
}

@inproceedings{sly2010computation,
  author    = {Sly, Allan},
  booktitle = {FOCS},
  title     = {Computational Transition at the Uniqueness Threshold},
  year      = {2010},
  publisher = {IEEE},
  volume    = {},
  number    = {},
  pages     = {287--296},
  doi       = {10.1109/FOCS.2010.34}
}

@article{sly2014computational,
  title     = {Counting in two-spin models on $d$-regular graphs},
  author    = {Sly, Allan and Sun, Nike},
  journal   = {Ann. Probab.},
  volume    = {42},
  number    = {6},
  pages     = {2383--2416},
  year      = {2014},
  publisher = {Institute of Mathematical Statistics},
  doi       = {10.1214/13-AOP888},
  url       = {https://doi.org},
  note      = {(Conference version in \emph{FOCS}'12)}
}

@inproceedings{Sri08,
  author    = {Srinivasan, Aravind},
  title     = {Improved algorithmic versions of the {L}ov\'{a}sz local lemma},
  booktitle = {SODA},
  pages     = {611--620},
  publisher = {SIAM},
  year      = {2008},
  mrclass   = {68W05 (05D40 68R05)},
  mrnumber  = {2487630}
}

@article{stefankovic2009adaptive,
  author   = {{\v{S}}tefankovi\v{c}, Daniel and Vempala, Santosh and Vigoda, Eric},
  title    = {Adaptive simulated annealing: a near-optimal connection
              between sampling and counting},
  journal  = {J. ACM},
  fjournal = {Journal of the ACM},
  volume   = {56},
  year     = {2009},
  number   = {3},
  pages    = {Art. 18, 36},
  issn     = {0004-5411},
  mrclass  = {68T05 (60C05 62D05)},
  mrnumber = {2536133},
  doi      = {10.1145/1516512.1516520},
  url      = {https://doi.org/10.1145/1516512.1516520}
}

@article{vishesh21sampling,
  author      = {Vishesh Jain and Huy Tuan Pham and Thuy-Duong Vuong},
  eprint      = {2102.08342},
  journal     = {arXiv preprint arXiv:2102.08342},
  title       = {On the sampling {L}ov{\'{a}}sz Local Lemma for atomic constraint satisfaction problems},
  eprintclass = {cs.DS},
  year        = {2021}
}

@inproceedings{vishesh21towards,
  author    = {Vishesh Jain and
               Huy Tuan Pham and
               Thuy-Duong Vuong},
  title     = {Towards the sampling {L}ov{\'{a}}sz Local Lemma},
  booktitle = {FOCS},
  pages     = {173--183},
  publisher = {IEEE},
  year      = {2021},
  doi       = {10.1109/FOCS52979.2021.00025}
}

@inproceedings{wang2024sampling,
  author    = { Wang, Chunyang and Yin, Yitong },
  author+an = {1=thesisauthor},
  booktitle = { FOCS},
  title     = { A Sampling {L}ovász Local Lemma for Large Domain Sizes },
  year      = {2024},
  volume    = {},
  issn      = {},
  pages     = {129--150},
  doi       = {10.1109/FOCS61266.2024.00019},
  url       = {https://doi.ieeecomputersociety.org/10.1109/FOCS61266.2024.00019},
  publisher = {IEEE}
}

@inproceedings{weitz06counting,
  author    = {Dror Weitz},
  booktitle = {{STOC}},
  pages     = {140--149},
  publisher = {ACM},
  title     = {Counting independent sets up to the tree threshold},
  year      = {2006}
}

@article{zhuk2020dichotomy,
  author    = {Dmitriy Zhuk},
  title     = {A Proof of the {CSP} Dichotomy Conjecture},
  journal   = {J. {ACM}},
  volume    = {67},
  number    = {5},
  pages     = {30:1--30:78},
  year      = {2020},
  publisher = {ACM},
  note      = {(Conference version in \emph{FOCS}'17)}
}

@misc{achlioptas2026sampling,
  author = {Dimitris Achlioptas},
  title = {A sampling {Lov\'asz} Local Lemma},
  year = {2026},
  note = {arXiv:2609.18712v2},
  url = {https://arxiv.org/abs/2609.18712v2}
}

@misc{jain2026general,
  author = {Vishesh Jain and Clayton Mizgerd and Huy Tuan Pham},
  title = {A general counting and sampling {Lov\'asz} local lemma},
  year = {2026},
  note = {arXiv:2609.23671v1},
  url = {https://arxiv.org/abs/2609.23671v1}
}

\appendix

\section{Verification of the Hardness Lower Bound}
\label{app:known-lower-bound}

Let $A=2^s$ for some integer $s\geq 1$,
$d=\lfloor\e A\rfloor$, and $\Delta=d+1$. By
\cite[Theorem~49]{BGGGS19}, $\#\mathrm{HyperIndSet}(2s,\Delta)$ is
NP-hard to approximate provided
\[
\frac{(\Delta-2)^\Delta}{(\Delta-1)^{\Delta-1}}
=(d-1)(1-d^{-1})^d
\ge \e^{-1}\left(d-\frac32+\frac1{6d}\right)
>A-1.
\]
The middle inequality follows directly from the power-series expansion
of $-\log(1-d^{-1})$. 
The construction puts $s$ variables at each graph vertex and one
clause on the two endpoint blocks per edge; thus
$p=A^{-2}$ and $D\le2(\Delta-1)=2d$.
Therefore,
\[
pD^2\le4(d/A)^2<4\e^2,
\]
which gives the stated known lower bound for approximate counting.

The sampling consequence follows from the usual self-reduction.
Order the clauses as $C_1,\ldots,C_m$, put
$F_i=\bigwedge_{j\le i}C_j$, and let
$q_i=\mu_{F_{i-1}}(C_i)$. Clause deletion does not increase $p$ or
$D$. Moreover, $F_{i-1}$ and $C_i$ are monotone in the same direction
under the product law, so they are positively correlated.
Thus $q_i\ge\Pr{C_i}=1-p\ge1/2$. Since
$Z(F_m)=2^n\prod_{i=1}^m q_i$, an approximate uniform sampler for the
deletion subinstances, with total-variation error
$O(\varepsilon/m)$, estimates every $q_i$ to sufficient accuracy
using polynomially many samples. Their product gives an FPRAS,
contradicting the counting hardness.

\end{document}